%% file: ergodicity-graphs.tex
\documentclass{amsproc}

\newtheorem{theorem}{Theorem}[section]
\newtheorem{corollary}[theorem]{Corollary}
\newtheorem{lemma}[theorem]{Lemma}

\theoremstyle{definition}
\newtheorem{definition}[theorem]{Definition}

\theoremstyle{remark}
\newtheorem{remark}[theorem]{Remark}

\numberwithin{equation}{section}

\newcommand{\abs}[1]{\lvert#1\rvert}

\begin{document}

\title[Quantum ergodicity on large regular graphs]{Quantum ergodicity on large regular graphs}

\author{Nalini Anantharaman}
\address{Universit\'{e} Paris-Sud 11, Math\'{e}matiques, B\^{a}t. 425, 91405 ORSAY
CEDEX, FRANCE}
\email{Nalini.Anantharaman@math.u-psud.fr}
\thanks{This material is based upon work supported by the National Science Foundation under agreement no. DMS 1128155, by the Fernholz foundation and by Agence Nationale de la Recherche under
the grant ANR-09-JCJC-0099-01. Any opinions, findings and conclusions or recommendations expressed in this material are those of the authors and do not necessarily reflect the views of the NSF}
\author{Etienne Le~Masson}
\address{Universit\'{e} Paris-Sud 11, Math\'{e}matiques, B\^{a}t. 425, 91405 ORSAY
CEDEX, FRANCE}
\email{Etienne.Lemasson@math.u-psud.fr}

\input{def.tex}

%    General info
\date{}


\begin{abstract}We propose a version of the Quantum Ergodicity theorem on large regular graphs of fixed valency. This is a property of delocalization of ``most'' eigenfunctions. We consider expander graphs with few short cycles (for instance random large regular graphs). Our method mimics the proof of Quantum Ergodicity on manifolds~: it uses microlocal analysis on regular trees, as introduced in \cite{LM}.
 \end{abstract}

\maketitle

\section{Introduction and main results}
It has been suggested by Kottos and Smilansky that graphs are a good ground of exploration of the ideas of ``quantum chaos'' \cite{KotSmi97,KotSmi99}. This means that the spectrum of the laplacian, as well as its eigenfunctions, should exhibit universal features that depend only on qualitative geometric properties of the graph. Whereas spectral statistics have been extensively studied, both numerically and analytically, the localization of eigenfunctions have (to our knowledge) only been investigated in a few models~: the {\em star} graphs (both metric and discrete) \cite{BKW04, KMW03}, the large {\em regular} discrete graphs \cite{Smi10, BL, Dumitriu}, and a family of metric graphs arising from measure preserving $1$-dimensional dynamical systems \cite{BKS07}. For the latter, a version of the ``Quantum Ergodicity theorem'' (also known as Shnirelman theorem) has been established. For star graphs, the paper
\cite{BKW04} shows on the opposite that ``Quantum Ergodicity'' holds neither in the high frequency limit nor in the large graph limit. Furthermore it shows there are eigenfunctions that localise on two bonds of the graph. Spectral properties of large regular \emph{discrete} graphs have been studied in \cite{MF91, LR96, JakMilRivR, Ter99, Smi07} but eigenfunctions have attracted attention only recently. A statistical study of the auto-correlations and the level sets of eigenvectors appeared in the papers \cite{Elon, ElonSmi} that introduce a random wave model (see also \cite{Gnu10} for a random wave model on metric graphs). The paper \cite{BL} has pioneered the study of quantum ergodicity on large regular graphs -- that is to say, the study of the spatial distribution of eigenfunctions of the laplacian. The result of \cite{BL} shows some form of delocalization of eigenfunctions~:

\begin{theorem}\cite{BL} \label{t:BL}Let $(G_n)$ be a sequence of $(q+1)$-regular graphs (with $q$ fixed), $G_n=(V_n, E_n)$ with $V_n =\{1,\ldots, n\}$.
Assume that\footnote{This assumption holds in particular if the injectivity radius is $\geq c \ln n$. The interest of the weaker assumption is that it holds for typical random regular graphs \cite{MKWW}.} there exists $c>0, \delta>0$ such that, for any $k\leq c\ln n$, for any pair of vertices $x, y\in V_n$,
$$|\{\mbox{paths of length } k \mbox{ in } G_n \mbox{ from } x \mbox{ to } y\}|\leq q^{k\left(\frac{1-\delta}2\right)}.$$

Fix $\eps>0$. Then, if $\phi$ is an eigenfunction of the discrete laplacian on $G_n$ and if $A\subset V_n$ is a set such that
$$\sum_{x\in A}|\phi(x)|^2\geq \eps \sum_{x\in V_n}|\phi(x)|^2,$$
then $|A|\geq n^\alpha$ --- where $\alpha>0$ is given as an explicit function of $\eps, \delta$ and $c$.
\end{theorem}

A similar form of delocalization (but on weaker scales) is established when the degree $q=q_n$ goes to infinity in \cite{Dumitriu, TranVu}.
We also refer to the papers \cite{ESY09, ESY09-2, EK11, EK11-2} where various forms of delocalization  
have been established for eigenvectors of random Wigner matrices and random band matrices.

In this paper, our aim is to establish for large regular graphs a result which reads like an analogue of the ``quantum ergodicity theorem'' on manifolds. Compared to Theorem \ref{t:BL} it pertains to a different definition of delocalization~: delocalization is now tested by averaging an observable and comparing with the average along the uniform measure. As a motivation, let us recall the Quantum Ergodicity theorem in its original form.

\bigskip

{\bf Quantum Ergodicity Theorem (Shnirelman theorem, \cite{Sni, CdV85, Zel87}).} Let $(M, g)$ be a compact Riemannian manifold.\footnote{We assume that $\Vol(M) = 1$.} Call $\Delta$ the Laplace-Beltrami operator on $M$. Assume that the geodesic flow of $M$ is {\em ergodic} with respect to the Liouville measure. Let $(\psi_n)_{n\in \IN}$ be an orthonormal basis of $L^2(M, g)$ made of eigenfunctions of the laplacian
$$\Delta\psi_n=-\lambda_n \psi_n,\qquad \lambda_n\leq \lambda_{n+1}\To +\infty.$$
Let $a$ be a continuous function on $M$ such that $\int_M a(x)d\Vol(x)=0$. Then
\begin{equation}\label{e:variance}\frac{1}{N(\lambda)}\sum_{n, \lambda_n\leq \lambda}|\la \psi_n, a\psi_n\ra_{L^2(M)}|^2 \Lim_{\lambda\To +\infty}0\end{equation}
where the normalizing factor is $N(\lambda)=|\{n, \lambda_n\leq \lambda\} |.$ Here $\la \psi_n, a\psi_n\ra_{L^2(M)}=\int_M a(x)|\psi_n(x)|^2d\Vol(x).$

\begin{remark}Equation \eqref{e:variance} implies that there exists a subset $S\subset \IN$ of density $1$ such that
\begin{equation}\label{e:subsequence}\la \psi_n, a\psi_n\ra \Lim_{n\To +\infty, n\in S} \int_M a(x)d\Vol(x).\end{equation}
(Note that \eqref{e:subsequence} is true even for a function $a$ with non-zero mean.)

In addition, since the space of continuous functions is separable, one can actually find $S\subset \IN$ of density $1$ such that
\eqref{e:subsequence} holds for {\em all} $a\in C^0(M)$. In other words, the sequence of measures $(|\psi_n(x)|^2d\Vol(x))_{n\in S}$ converges weakly to the uniform measure $d\Vol(x)$.

Actually, the full statement of the theorem says that
there exists a subset $S\subset \IN$ of density $1$ such that
\begin{equation}\label{e:subsequenceOp}\la \psi_n, A\psi_n\ra \Lim_{n\To +\infty, n\in S} \int_{S^*M} \sigma^0(A) dL\end{equation} 
for every pseudodifferential operator $A$ of order $0$ on $M$. On the right-hand side, $\sigma^0(A)$ is the principal symbol of $A$, that is a function on the unit cotangent bundle $S^*M$, and $L$ is the normalized Liouville measure (uniform measure), arising naturally from the symplectic structure of the cotangent bundle.
\end{remark}

\bigskip

Here we consider a sequence of $(q+1)$-regular connected graphs $(G_n)_{n\in\IN}$, $G_n=(V_n, E_n)$ with vertices $V_n =\{1,\ldots, n\}$ and edge set $E_n$. The valency $(q+1)$ is fixed. We denote by ${\mathfrak{X}}$ the $(q+1)$-regular tree and identify it with its set of vertices, equipped with the geodesic distance $d_{\mathfrak{X}}$, that we will simply denote by $d$ most of the time. We will denote by $d_{G_n}$ the geodesic distance on the graph. Each $G_n$ is seen as a quotient of ${\mathfrak{X}}$ by a group of automorphisms $\Gamma_n$~: $G_n= \Gamma_n\backslash {\mathfrak{X}}$, where we assume that the elements of $\Gamma_n$ act without fixed points. Accordingly, a function $f: V_n\To \C$ may be seen as a function $f: {\mathfrak{X}}\To \C$ satisfying $f(\gamma\cdot x)=f(x)$ for each $x\in {\mathfrak{X}}, \gamma\in \Gamma_n$. We will denote by $D_n$ a subtree of ${\mathfrak{X}}$ which is a fundamental domain for the action of $\Gamma_
n$ on vertices.

We consider the stochastic operator acting on $\Gamma_n$-invariant functions,
\begin{equation}\label{e:laplacian}
 A f(x)=\frac1{q+1}\sum_{{x\sim y} } f(y)
\end{equation}
where $x\sim y$ means that $x$ and $y$ are neighbours in the tree\footnote{This is also the (normalized) adjacency matric of the graph $G_n$, but note that this definition allows $G_n$ to have loops and multiple edges.}. It is related to the discrete laplacian by $A-I=\Delta.$

Whereas the Shnirelman theorem deals with the {\em high frequency} asymptotics ($\lambda_n\To +\infty$), there
is no such asymptotic r\'egime for discrete graphs since the laplacian is a bounded operator. We will instead work (like in \cite{BL}) in the {\em large spatial scale} r\'egime $n\To +\infty$.

We will assume the following conditions on our sequence of graphs~:

\bigskip

(EXP) The sequence of graphs is a family of expanders. More precisely, there exists $\beta>0$ such that the spectrum of $A$ on $L^2(G_n)$ is contained in $\{1\}\cup [-1+\beta, 1-\beta]$ for all $n$.

(EIIR) For all $R$, $\frac{|\{x\in V_n, \rho(x)<R\}|}{n}\To 0$
where $\rho(x)$ is the injectivity radius at $x$ (meaning the largest $\rho$ such that the ball $B(x, \rho)$ in $G_n$ is a tree).

(EIIR) is equivalent to saying that there exists $R_n\To +\infty$ and $\alpha_n\To 0$ such that
$$\frac{|\{x\in V_n, \rho(x)<R_n\}|}{n}\leq \alpha_n.$$
In particular, it is satisfied if the injectivity radius goes to infinity (with $R_n$ taken to be the minimal injectivity radius and $\alpha_n=0$).

Condition (EXP) replaces the ergodicity assumption in the usual quantum ergodicity theorem.

\bigskip

\begin{ex}\label{ex:random} 
The graph $G_n$ can be chosen uniformly at random among the $(q+1)$-regular graphs with $n$ vertices (see \cite{Bol01} section 2.4 for an introduction to this model). We can then take $R_n= k$ and
$n\alpha_n=40 A kq^k$ for any $k=k(n)$ such that $kq^k n^{-1} \Lim_{n\to +\infty} 0$, and $A=A(n)$ such that $A\geq c>1$ (see \cite{MKWW}, Theorem 4).
For instance, we can take $k=(1-\delta)\log_q (n)$, with $0 < \delta < 1$, and $A=2$. In this case we have $R_n = (1-\delta)\log_q (n)$ and $\alpha_n=80 (1-\delta) \log_q(n) n^{-\delta}$. For this choice of parameters, (EIIR) is satisfied with a probability tending to $1$ when $n \To +\infty$. More precisely, this probability is greater than $1 - e^{-C n^{1-\delta}}$, for some constant $C > 0$ independent of $n$.

Condition (EXP) is also satisfied by these sequences of random graphs~: \cite{Alon} proves an equivalence between having a uniform spectral gap and having a uniform Cheeger constant. The latter condition was shown to hold generically in \cite{Pinsker}. In \cite{Fri08}, a spectral gap estimate that is close to optimal is established.
\end{ex}

\begin{ex}
 An explicit example of sequence of $(q+1)$-regular graphs to which our results apply is given by the construction of Ramanujan graphs of \cite{LPS88} for prime $q$. The sequence obtained satisfies conditions (EXP) and (EIIR) even more strongly than the sequences of random graphs of Example \ref{ex:random}. A method for obtaining bi-partite Ramanujan graphs of arbitrary degrees has appeared recently in \cite{MSS}.
\end{ex}

\bigskip

Eigenvalues of $A$ on a $(q+1)$-regular graph may be parameterized by their ``spectral parameter''  $s$ thanks to the relation
\begin{equation}\label{e:parametrization}
 \lambda=\frac{2\sqrt q}{q+1}\cos (s\ln q).
\end{equation}
The case $s\in \IR$ corresponds to $\lambda\in \left[-\frac{2\sqrt q}{q+1}, \frac{2\sqrt q}{q+1}\right]$, which is the tempered spectrum. In that case we will usually choose $s\in [0, \tau]$ ($\tau=\frac{\pi}{\ln (q)}$). The case $s\in i(-1/2, 1/2)+ik\frac{\pi}{\ln (q)}$ ($k\in \IZ$) corresponds to $\lambda\in [-1, 1]\setminus\left(-\frac{2\sqrt q}{q+1}, \frac{2\sqrt q}{q+1}\right)$, which is the untempered spectrum. The result of this paper will only be of interest in the tempered part of the spectrum.

In what follows, $(r_n)$ will be a sequence satisfying $r_n+2\leq R_n$ and $q^{r_n}\alpha_n \To 0$. The sequence $(\delta_n)$
will be assumed to satisfy $\delta_n^K r_n^{K-4}\To +\infty$, for some integer $K$. We also assume that $\delta_n \To 0$, although it is not necessary for the general proof of section \ref{s:general}.

Our aim is to prove the following~:

\begin{theorem} \label{t:main} Let $(G_n)$ be a sequence of $(q+1)$-regular graphs, $G_n=(V_n, E_n)$ with $V_n =\{1,\ldots, n\}$. Assume that $(G_n)$ satisfies (EIIR) and (EXP). Fix $s_0\in (0, \tau)$ and let $I_n=[s_0-\delta_n, s_0+\delta_n]$. Call $(s^{(n)}_1,\ldots, s^{(n)}_n)$ the spectrum of $A$ on $G_n$, and $(\psi^{(n)}_1,\ldots, \psi^{(n)}_n)$ a corresponding orthonormal eigenbasis.

Let $N(I_n, G_n)=\left|\{j\in \{1,\ldots, n\}, s^{(n)}_j\in I_n\}\right|$ be the number of eigenvalues in $I_n$.\footnote{Note that with the assumptions on $\delta_n$, $N(I_n,G_n) \To +\infty$ (see Corollary \ref{c:weyl}).} Finally, let $a_n : V_n \To \C$ be a sequence of functions such that 

$$\sum_{x\in V_n} a_n(x)=0, \qquad   \sup_x |a_n(x)|\leq 1.$$

Then $$\frac1{N(I_n, G_n)} \sum_{s^{(n)}_j\in I_n} \left| \la \psi^{(n)}_j, a_n\psi^{(n)}_j\ra\right|^2\Lim_{n\To +\infty} 0.$$
\end{theorem}

\begin{remark} If $a_n$ does not have zero mean, then by applying the theorem to $a_n- \overline{a_n}$ (where $\overline{a_n}=\frac1n
\sum_{x\in V_n} a_n(x)$) we obtain
$$\frac1{N(I_n, G_n)} \sum_{s^{(n)}_j\in I_n} \left| \la \psi^{(n)}_j, a_n\psi^{(n)}_j\ra - \overline{a_n}\right|^2\Lim_{n\To +\infty} 0.$$
\end{remark}

\begin{remark}\label{r:weakexp}
  If we exclude the case $s_0 = \tau/2$ in Theorem \ref{t:main}, we can assume, instead of (EXP), the following weaker condition~: there exists $\beta>0$ such that the spectrum of $A$ on $L^2(G_n)$ is contained in $\{1\}\cup [-1, 1-\beta]$ for all $n$. In particular, the theorem applies for bipartite regular graphs in this case.
  
  We can also say something in the case $s_0 = \tau/2$ for bipartite expander graphs, that is if there exists $\beta>0$ such that the spectrum of $A$ on $L^2(G_n)$ is contained in $\{-1,1\}\cup [-1+\beta, 1-\beta]$ for all $n$. We need to strengthen the condition on the functions $a_n(x)$ in the theorem for the conclusion to apply~: if $V_n=V_n^1\sqcup V_n^2$ is the bi-partition of $V_n$, then we need that \[ \sum_{x\in V_n^1} a_n(x)= \sum_{x\in V_n^2} a_n(x) = 0. \]
  The theorem then tells us that we have equidistribution of most eigenfunctions with eigenvalue near $\tau/2$ on each set $V_n^1$ and $V_n^2$, without providing information on the relative weight of these two sets. 
\end{remark}

\begin{remark} The proof will show that we can weaken condition (EXP), by allowing the spectral gap $\beta$ to decay with $n$ ``not too fast'' ($\beta \gg r_n^{-2/9}$ is enough). See Section \ref{s:quantitative}.
\end{remark}

Since random $(q+1)$-regular graphs satisfy both (EXP) \cite{Pinsker, Alon,Fri08} and (EIIR) \cite{MKWW}, our theorem applies to them with the values of $R_n, \alpha_n$ given in Example \ref{ex:random}.
\begin{coro} Let $a_n : V_n=\{1, \ldots, n\} \To \C$ be a sequence of functions such that 

$$\sum_{x\in V_n} a_n(x)=0, \qquad  \sup_x |a_n(x)|\leq 1.$$

Choose $(G_n)$ uniformly at random amongst  the $(q+1)$-regular graphs $G_n=(V_n, E_n)$ such that $V_n = \{1,\ldots, n\}$. Choose $j$ uniformly at random in $N(I_n, G_n)$.

Then for any fixed $\eps>0$, \[P\left( \left| \la\psi_j^{(n)}, a_n\psi_j^{(n)}\ra \right|\geq \eps\right) \To  0\]
when $n\To +\infty$.
\end{coro}

The statement of Theorem \ref{t:main} is the exact analogue of the Shnirelman theorem in its form \eqref{e:variance}. 
However, we do not have a statement analogous to the convergence of measures \eqref{e:subsequence}, because our sequence of measures does not live on a single space; instead, it is defined on the sequence of graphs $G_n$. We do not know of a notion that would be adapted to describe the limit of the family $(G_n)$ endowed with the probability measure $(|\psi^{(n)}_j(x)|^2)_{x\in V_n}.$

We can generalize Theorem \ref{t:main} by replacing the function $a$ with any finite range operator~:
\begin{theorem} \label{t:gen} Let $(G_n)$ be a sequence of $(q+1)$-regular graphs, $G_n=(V_n, E_n)$ with $V_n =\{1,\ldots, n\}$. Assume that $(G_n)$ satisfies (EIIR) and (EXP). Fix $s_0\in (0, \tau)$ and let $I_n=[s_0-\delta_n, s_0+\delta_n]$. Call $(s^{(n)}_1,\ldots, s^{(n)}_n)$ the spectrum of the laplacian on $G_n$, and $(\psi^{(n)}_1,\ldots, \psi^{(n)}_n)$ a corresponding orthonormal eigenbasis.

Let $N(I_n, G_n)=\left|\{j\in \{1,\ldots, n\}, s^{(n)}_j\in I_n\}\right|$. 

Fix $D\in \N$. Let $A_n$ be a sequence of operators on $L^2(G_n)$ whose kernels $K_n : V_n \times V_n\To \C$ are such that $K_n(x, y)=0$ for $d(x, y)>D$. 

Assume that
$    \sup_{x, y\in V_n}|K_n(x, y)| \leq 1.$

Then there exists a number $\overline{A_n}(s_0)$ such that
$$\frac1{N(I_n, G_n)} \sum_{s^{(n)}_j\in I_n} \left| \la \psi^{(n)}_j, A_n\psi^{(n)}_j\ra -\overline{A_n}(s_0)
\right|^2\Lim_{n\To +\infty} 0.$$

With the notation of \S \ref{s:class}, we can write $A_n=\Op(a_n)$, $a_n\in S_o^D$; and we have the expression $\overline{A_n}(s_0)
=\frac1n\sum_{x\in D_n}\int_{\Omega} a_n(x, \omega, s_0) d\nu_x(\omega)$.
\end{theorem}

Quantitative statements (i.e. rates of convergence) will be given in Section \ref{s:quantitative}.

\section{Theorem \ref{t:main}~: outline of the proof in the case $s_0=\tau/2$}\label{s:altern}
We first give a proof in the special case $s_0=\tau/2$. The reason for treating this case separately is that one can give a proof which is exactly parallel to that of the Shnirelman theorem on manifolds \cite{Sni, CdV85, Zel87, ZelC}. The case of arbitrary $s_0 $ requires additional arguments and will be treated in Section \ref{s:general}.

\subsection{Upper bound on the variance\label{s:qvar}}  

Fix an integer $T>0$. Let $\chi$ be a smooth cut-off function supported in $[-1, 1]$ and taking the constant value $1$ on $[-1/2, 1/2]$. We write
\begin{equation}\chi_n(s)=\chi\left(\frac{s-s_0}{2\delta_n}\right)\label{e:chin}\end{equation}
so that $\chi_n\equiv 1$ on $I_n$. We use the pseudodifferential calculus and the notation defined in Section \ref{s:PDC}, taking the cut-off parameter $r$ equal to $r_n$ (from condition (EIIR), as explained before the statement of Theorem \ref{t:main}).

To simplify the notation, we will write $\psi_j = \psi_j^{(n)}$, $s_j = s_j^{(n)}$, and $a = a_n$. The observable $a$ is a function on $G_n$, in other words a $\Gamma_n$-invariant function on ${\mathfrak{X}}$. Let $\Omega$ be the boundary of ${\mathfrak{X}}$ (see section \ref{s:PDC}), then $a$ extends to a function on ${\mathfrak{X}}\times \Omega\times [0, \tau]$ that does not depend on the last two coordinates. The notation $\Op(a)$ is then defined in section \ref{s:PDC}.

Thanks to Lemma \ref{l:chin}
and to the ``Egorov property'' Corollary \ref{c:egorovtau2}, we have\footnote{To prove the extended Theorem \ref{t:gen}, we also need Lemma \ref{l:produit}.}
\begin{eqnarray*}
\sum_{j=0}^n \chi_n(s_j)^2 \left| \la \psi_j, a\psi_j\ra \right|^2 
  &=&  \sum_{j=0}^n \left| \la \psi_j, \Op_{G_n}(a\chi_n)\psi_j\ra\right|^2 +n r^3 O\left(\frac{1}{(r\delta_n)^\infty} \right)\\
  &=& \sum_{j=0}^n \left| \la \psi_j, \Op_{G_n}(a^T\chi_n)\psi_j\ra\right|^2+n r^3 O\left(\frac{1}{(r\delta_n)^\infty} \right)\\
  &&+ O\lp \frac{T^2}{r^2} \rp n \int \chi_n^2(s)d\mu(s) \\&&+ O(T^2 q^{r+2}) |\{x\in D_n, \rho(x)\leq r + 2\}|\int \chi_n^2(s)d\mu(s) \\
  &&+ O(T^2) n \int O(|s-s_0|^2) \chi_n^2(s)d\mu(s) 
\end{eqnarray*}
where 
$a^T:=\frac1T\sum_{k=0}^{T-1} a\circ \sigma^{2k}$ and $\sigma: {\mathfrak{X}}\times\Omega\To {\mathfrak{X}}\times\Omega$ is the shift (see \S \ref{s:Egorov}).

Next, we use Lemma \ref{l:HSnorm} to write
\begin{eqnarray*} \sum_{j=0}^n \left| \la \psi_j, \Op_{G_n}(a^T\chi_n)\psi_j\ra\right|^2&\leq &
\norm{\Op_{G_n}(a^T\chi_n)}^2_{HS}\\
&\leq &\sum_{x\in D_n}\int |a^T(x, \omega)|^2 d\nu_x(\omega)\chi_n(s)^2 d\mu(s)\\&& + q^r \sum_{x\in D_n, \rho(x)\leq r}\int |a^T(x, \omega)|^2 d\nu_x(\omega) \chi_n(s)^2d\mu(s)\\
&\leq &\sum_{x\in D_n}\int |a^T(x, \omega)|^2 d\nu_x(\omega) \int \chi_n(s)^2 d\mu(s) \\
&&+ q^r |\{x\in D_n, \rho(x)\leq r\}| \int \chi_n(s)^2 d\mu(s).
\end{eqnarray*}
 We also know from the Kesten-McKay law (Section \ref{s:weyl}, Corollary \ref{c:weyl}) that
 $$ \frac{n}{N(I_n, G_n)}\int\chi_n(s)^2 d\mu(s)=O(1).$$
 We thus have
 \begin{align*}
  \frac1{N(I_n, G_n)} \sum_{s_j\in I_n} \left| \la \psi_j, a\psi_j\ra\right|^2
  &= O(1) \frac{1}{n}\sum_{x\in D_n}\int |a^T(x, \omega)|^2 d\nu_x(\omega) \\
  &\quad + O(T^2 q^{r+2}) \frac{|\{x\in D_n, \rho(x)\leq r + 2\}|}{n} \\
  &\quad + \frac{n}{N(I_n, G_n)} r^3 O\left(\frac{1}{(r\delta_n)^\infty} \right) + O\lp\frac{T^2}{r^2}\rp + O(T^2 \delta_n^2)
  \\ &= O(1) \frac{1}{n}\sum_{x\in D_n}\int |a^T(x, \omega)|^2 d\nu_x(\omega) \\
  &\quad + O(T^2 q^r) \alpha_n + O(1) \frac{r^3}{\delta_n} O\left(\frac{1}{(r\delta_n)^\infty} \right)\\
  &\quad  + O\left(\frac{T^2}{r^2}\right) + O(T^2 \delta_n^2) .
 \end{align*}
 Our choices of $r=r_n$ and $\delta_n$ imply that the last four terms vanish as $n$ goes to infinity while $T$ is fixed.
 
\subsection{Expansion and ergodicity}\label{s:ergodicity}

We write, using \eqref{e:adjoint}
 \begin{align*}
 \frac1n \sum_{x\in D_n}\int |a^T(x, \omega)|^2 d\nu_x(\omega)
 &= \frac1n \frac{1}{T^2}\sum_{k=0}^{T-1}\sum_{j=0}^{T-1} \sum_{x\in D_n}\int a\circ \sigma^{2|k-j|}(x, \omega) a(x) d\nu_x(\omega)
% \\&=\frac1n \frac{1}{T^2}\sum_{k=0}^{T-1}\sum_{j=0}^{T-1} \sum_{x\in D_n}S_{2|k-j|}a (x)a(x)
 \end{align*}
 For every $x \in \mathfrak{X}$ and $k \in \N$, define the partition of $\Omega$ in $(q+1)q^{k-1}$ sets\footnote{See section \ref{s:PDC} for a definition of the boundary $\Omega$ and of the notation $[x,\omega)$.}
 \[\Omega(x,y) := \left\{ \omega \in \Omega \: | \: [x,\omega) = (x,x_1,x_2,\ldots) \text{ and } x_k = y \right\}, \]
  where $y \in \mathfrak{X}$ and $d(x,y) = k$. 
 Then $\omega \mapsto a\circ\sigma^k(x,\omega)$ is constant on $\Omega(x,y)$ for every $y \in \mathfrak{X}$ such that $d(x,y) = k$, and $\nu_x(\Omega(x,y)) = \frac1{(q+1)q^{k-1}}$. We have
 \begin{align*}
  \sum_{x\in D_n}\int_\Omega a\circ \sigma^{k}(x, \omega) a(x) d\nu_x(\omega) &= \sum_{x\in D_n}\sum_{\substack{y\in\mathfrak{X}\\d_{\mathfrak{X}}(x,y)=k}} \int_{\Omega(x,y)} a\circ \sigma^{k}(x, \omega) a(x) d\nu_x(\omega)\\
  &=\sum_{x\in D_n}S_{k}a(x)a(x),
 \end{align*}
 where $S_0 = \text{Id}$ and for all $k \geq 1$, $S_k$ is the stochastic operator defined as follows by its kernel on the tree ${\mathfrak{X}}$~:
 \begin{equation}\label{e:sk}S_k f (x)= \frac{1}{(q+1)q^{k-1}}\sum_{d_{\mathfrak{X}}(x, y)=k} f(y).\end{equation}
 We thus have
 \begin{align*}
 \frac1n \sum_{x\in D_n}\int |a^T(x, \omega)|^2 d\nu_x(\omega)
 &=\frac1n \frac{1}{T^2}\sum_{k=0}^{T-1}\sum_{j=0}^{T-1} \sum_{x\in D_n}S_{2|k-j|}a (x)a(x).
 \end{align*}
 
 On the quotient $G_n$, the spectrum of $S_k$ is the set $\left\{  \phi_{s_j}(k), j=1, \ldots, n \right\}  $  
 where $\phi_s$ is the spherical function,
 \begin{equation}\label{e:spherical}
   \phi_s(k)=q^{-k/2}\left(\frac2{q+1}\cos (ks\ln(q)) + \frac{q-1}{q+1}\frac{\sin (k+1)s \ln(q)}{\sin s \ln(q)} \right)
 \end{equation}
 and $\{ s_j, j=1, \ldots, n \}$ are the spectral parameters for the operator $A$ defined by \eqref{e:laplacian}.
 
 Using the parameterization \eqref{e:parametrization} of the spectrum, the eigenvalue $\lambda=1$ corresponds to
 $$is  = \frac12.$$
 %$$is   + i \frac{\pi}{\ln q}=   \frac12$$
 Because of the (EXP) condition, the other untempered eigenvalues satisfy $is\in (0, \frac12-\beta)$ or $is   + i \frac{\pi}{\ln q}\in (0, \frac12-\beta)$, for some $\beta>0$ independent of $n$. It follows that $|\phi_{s_j}(k)| \leq C q^{-\beta k}$ with $C,\beta$ independent of $n$. The eigenvalues of the self-adjoint stochastic operator
 $\frac{1}{T^2}\sum_{k=0}^{T-1}\sum_{j=0}^{T-1} S_{2|k-j|}$ are therefore bounded by 
 \[ \frac{C}{T^2}\sum_{k=0}^{T-1}\sum_{j=0}^{T-1} q^{-2\beta|k-j|} \leq \frac{C}{2T(1-q^{-2\beta})}  \]
 in modulus. They are contained in $\{1\}\cup \left[-\frac{C}{T\beta}, \frac{C}{T\beta}\right]$ for some $C$, independent of $n, T, \beta$ (the eigenvalue $1$ has multiplicity $1$, corresponding to the constant function).
 
 Thus, if $a$ satisfies $\sum_{x\in V_n}a(x)=0$ and $\sup_x |a(x)|\leq 1$, we have
  $$\frac1n \frac{1}{T^2}\sum_{k=0}^{T-1}\sum_{j=0}^{T-1} \sum_{x\in D_n}S_{2|k-j|}a (x)a(x)=O\left(\frac1{T\beta}\right).$$

 \subsection{Conclusion}
  We obtain, using the results of the previous sections and the Kesten-McKay law (Corollary \ref{c:weyl}),
  \begin{align*}
   \frac1{N(I_n, G_n)} \sum_{s_j\in I_n} \left| \la \psi_j, a\psi_j\ra\right|^2 
   &\leq O\left(\frac{r^3}{\delta_n(r\delta_n)^\infty} \right)+ O\left(\frac{T^2}{r^2}\right) + O(T^2\delta_n^2) \\
   &\quad+ O(T^2q^r\alpha_n) + O\left(\frac1{T\beta}\right).
  \end{align*}
 If we choose the sequences $r = r_n$ and $\delta_n$ satisfying $q^r\alpha_n \To 0$ and  $\frac{r^3}{\delta_n(r\delta_n)^K}\To 0$ for some integer $K $, we finally have
 $$ \limsup_{n \to \infty} \frac1{N(I_n, G_n)} \sum_{s_j\in I_n} \left| \la \psi_j, a\psi_j\ra\right|^2= O\left(\frac1{T\beta}\right).$$
 As the left-hand side of the equality does not depend on $T$, we take the limit $T \To \infty$ to obtain
 $$ \lim_{n \to \infty} \frac1{N(I_n, G_n)} \sum_{s_j\in I_n} \left| \la \psi_j, a\psi_j\ra\right|^2=0. $$

\section{Elements of pseudodifferential calculus}\label{s:PDC}
In \S \ref{s:tree}--\ref{s:class} we recall some of the tools of pseudodifferential calculus that were introduced in \cite{LM}. However, the following important remark has to be made~: in order for Theorems \ref{t:main} and \ref{t:gen} to have full strength, we should not impose on the symbols $a$ too strong regularity conditions, that would have the effect of making the theorems trivial consequences of the (EXP) condition. Thus, we pay attention to only use from \cite{LM} the properties that do not require regularity of $a$ with respect to the $x$-variable.

In the following sections we try to construct a pseudodifferential calculus on the quotient.

\subsection{Definition of $\Op(a)$ on the infinite $(q+1)$-regular tree.\label{s:tree}}
Let $\Omega$ be the boundary of the tree. It is the set of equivalence classes of infinite half-geodesics of ${\mathfrak{X}}$ for the relation~: two half-geodesics $(x_1,x_2,x_3,\ldots)$ and $(y_1,y_2,y_3,\ldots)$ are equivalent iff there exist $k, N \in \N$ such that for all $n \geq N$, $x_{n+k} = y_n$. For every $\omega \in \Omega$, we will denote by $[x,\omega)$ the unique half-geodesic starting at $x$ and equivalent to $\omega$.

Let $a: {\mathfrak{X}}\times\Omega\times [0, \tau]\To \C$ be a bounded measurable function.
In \cite{LM}, the operator $\Op_{\mathfrak{X}}(a)$ was defined by
$$\Op_{\mathfrak{X}}(a) u(x)=\sum_{y\in {\mathfrak{X}}}\int_{\Omega}\int_0^\tau q^{\left(\frac12+is\right)(h_\omega(y)-h_\omega(x))}a(x, \omega, s)u(y) d\nu_x(\omega) d\mu(s)$$
for every $u:{\mathfrak{X}}\to \C$ with finite support. Here $h_\omega(x)$ is the height function (or Busemann function), $d\mu(s)$ is the Plancherel measure associated to the $(q+1)$-regular tree\footnote{$d\mu(s)=|c(s)|^{-2}ds$ where $c$ is the Harish-Chandra function.}, and $\nu_x$ is the harmonic measure on $\Omega$, seen from the point $x$. We refer to \cite{CS99} for more background.
We will denote by
$$K_{\mathfrak{X}}(x, y;a)=\int_{\Omega}\int_0^\tau q^{\left(\frac12+is\right)(h_\omega(y)-h_\omega(x))}a(x, \omega, s)  d\nu_x(\omega) d\mu(s)$$
the kernel of $\Op_{\mathfrak{X}}(a)$.

\subsection{Class of symbols\label{s:class}}
From \cite{CS99}, we know that the fact that $K_{\mathfrak{X}}(x, y;a)=0$ for $d_{\mathfrak{X}}(x, y)>D$ is equivalent to the four following conditions on $a$~:

\begin{itemize}
\item $a$ is a continuous function;
\item $a$ extends to a $2\tau$-periodic entire function of exponential type $D$ uniformly in $\omega$; i.e. for all $x$ there exists $C(x)>0$ such that
$$|a(x, \omega, z)|\leq C(x) q^{D|\Im m(z)|}\qquad \forall \omega\in\Omega, \forall x\in {\mathfrak{X}};$$
\item $a$ satisfies the symmetry condition
$$\int_{\Omega}q^{(\frac12-is)(h_\omega(y)-h_\omega(x))}a(x, \omega, s)d\nu_x(\omega)=\int_{\Omega}q^{(\frac12+is)(h_\omega(y)-h_\omega(x))}a(x, \omega, -s)d\nu_x(\omega)$$
for all $s\in\C$, $x\in {\mathfrak{X}}$.
\item $a$ is a $D$-cylindrical function, that is~: if the two half-geodesics $[x, \omega)=(x_0, x_1, x_2,\ldots)$ and  $[x', \omega')=(x'_0, x'_1, x'_2,\ldots)$ satisfy $x_j=x'_j$ for $0\leq j\leq D$, then $a(x, \omega, s)=a(x', \omega', s)$.
\end{itemize}

We shall denote by $S_o^D({\mathfrak{X}})$ the class of such functions. In \cite{LM}, another class of symbols was considered~:

\begin{definition} $S({\mathfrak{X}}) $ is the class of functions $a: {\mathfrak{X}}\times\Omega\times [0, \tau]\To \C$ such that
\begin{itemize}
\item for every $(x, \omega, s)$, for every $k \in \N$, $\partial_s^k a(x, \omega, s)$ is continuous on $\mathfrak{X}\times\Omega\times [0, \tau]$, and for every $l \in \N$, there exists $C_l>0$ such that, for all $n \in \N$, for every $(x, \omega, s)$, $$\left| (a-\mathcal{E}_n^x a)(x, \omega, s)\right|\leq\frac{C_l}{(1+n)^l}.$$
\item for every $k \in \N$, and every $(x, \omega)$, $\partial_s^k a(x, \omega, 0)=\partial_s^k a(x, \omega, \tau)=0.$ 
\end{itemize}
\end{definition}
Here $\mathcal{E}_n^x a(x,\omega,s)$ is the projection of $a(x,\omega,s)$ on functions depending only on the first $n$ vertices of the half-geodesic $[x,\omega)$ (see \cite{LM} for a formula). In particular, $\mathcal{E}_n^x a(x,s) = a(x,s)$ if $a$ does not depend on $\omega$.

It is proven in \cite{LM} that $S({\mathfrak{X}})$ endowed with usual addition and multiplication is an algebra. This makes it more suitable for semiclassical analysis than the class $S_o^D({\mathfrak{X}})$. It also has the property, crucial for us, that
$$a\in S({\mathfrak{X}})\Longrightarrow a\circ\sigma\in S({\mathfrak{X}})$$
where $\sigma$ is the shift, $\sigma(x, \omega)=(x_1, \omega)$ if $[x, \omega)=(x_0, x_1, x_2,\ldots)$.
It is proven in \cite{LM} that
 \begin{equation}\label{e:decay}
  |K_{{\mathfrak{X}}}(x,y;a)| \leq C \left(\norm{a}_{\Omega,M} + \sum_{k=0}^{M+1} \norm{\partial_s^k a}_\infty \right) \frac{q^{-\frac{d(x,y)}{2}}}{(1+d(x,y))^M}
 \end{equation}
 for any $M$, where $\norm{a}_{\Omega,M} = \sup_{(x,\omega,s)} \sup_n (1+n)^M |a-\mathcal{E}_n^x a|(x,\omega,s) $, and that as a consequence, $\Op_{\mathfrak{X}}(a)$ extends to a bounded operator on $L^2({\mathfrak{X}})$ if $a\in S({\mathfrak{X}})$.
 
If $a(x, \omega, s)=a(x)$ depends only on $x$, then $\Op_{\mathfrak{X}}(a)$ is the operator of multiplication by $a$. At several places we will use the fact that
 \begin{equation}\Op_{\mathfrak{X}}(a)\Op_{\mathfrak{X}}(\varphi)=\Op_{\mathfrak{X}}(a\varphi)\label{e:obvproduct}\end{equation}
 if $\varphi=\varphi(s)$ only depends on the last variable and $a(x, \omega, s)\in S(\mathfrak{X})$, say.

 In most of what follows, we will actually need very few conditions on the symbols $a(x,\omega,s)$. Essentially it will be required that $a \in L^\infty$ and that for every $x \in {\mathfrak{X}}$, $a(x,\cdot,\cdot) \in L^2(\Omega \times [0,\tau])$. For convenience, we can assume that $a \in S_o^D({\mathfrak{X}})$ for some $D \in \N$, or that $a \in S({\mathfrak{X}})$. In Lemma \ref{l:produit} we use the condition  $a \in S_o^D({\mathfrak{X}})$.
 
\subsection{Definition of $\Op_{G_n}(a)$ on a finite graph.} Recall that $G_n$ is written as a quotient $\Gamma_n\backslash {\mathfrak{X}}$, where $\Gamma_n$ is a group of automorphisms of ${\mathfrak{X}}$, whose elements act without fixed points.

Let us now assume that $a$ is $\Gamma_n$-invariant, meaning that $a(\gamma\cdot x, \gamma\cdot\omega, s)=a(x, \omega, s)$ for all $(x, \omega, s)$ and all $\gamma\in\Gamma_n$ (where the action of $\Gamma_n$ on the boundary $\Omega$ is obtained by extending its action on ${\mathfrak{X}}$). For a $\Gamma_n$-invariant symbol, we have
$$K_{\mathfrak{X}}(\gamma\cdot x,\gamma\cdot y;a)= K_{{\mathfrak{X}}}(x,y;a)$$
for all $x, y\in {\mathfrak{X}}$ and $\gamma\in\Gamma_n$. The proof of this fact is identical to the proof of Proposition 1.1 in \cite{Zel86}.

We now define $\Op(a)$ on the quotient.

\begin{definition} \label{d:maindef}Assume the sequence $(G_n)$ satisfies (EIIR). Let $r=r_n$ be a positive number.

If $a$ is $\Gamma_n$-invariant, we define $\Op_{G_n}(a)$ to be the operator with $\Gamma_n$-bi-invariant kernel $$K_{G_n}(x, y;a)=\sum_{\gamma\in \Gamma_n}K_{\mathfrak{X}}(x,\gamma\cdot y;a) \chi\left(\frac{d(x,\gamma\cdot y)}{r}\right).$$
\end{definition}

Here $\chi$ is a cut-off function that satisfies the conditions of \S \ref{s:qvar} (although it need not be the same cut-off as in \S \ref{s:qvar}, we use the same notation).

Compared to the case of manifolds, a difficulty we meet is that we are not able to prove that  $\Op_{G_n}(a)$ is bounded on $L^2(V_n)$ independently of $n$ (actually, inspection of simple examples show that our conditions on $a$ are not sufficient to ensure this). Note however that we are only interested in $\Op_{G_n}(a)\psi^{(n)}_j$ for $\lambda^{(n)}_j$ in the tempered spectrum; more precisely, we shall only need to estimate quantities such as
$\frac1{N(I_n, G_n)} \sum_{s^{(n)}_j\in I_n} \left| \la \psi^{(n)}_j, \Op_{G_n}(a)\psi^{(n)}_j\ra\right|^2.$
For that purpose it will be sufficient to know that the Hilbert-Schmidt norm of $\Op_{G_n}(a)$ does not grow too fast~:

\begin{lemma} \label{l:HSnorm}We have for every $r \geq 0$
\begin{multline*}
\norm{\Op_{G_n}(a)}^2_{HS} 
\leq \sum_{\substack{x\in D_n \\ \rho(x)\geq r}}\sum_{\substack{y\in {\mathfrak{X}} \\ d(y, x)\leq r}}|K_{\mathfrak{X}}(x, y; a)|^2 + q^r \sum_{\substack{x\in D_n \\ \rho(x)\leq r}}\sum_{\substack{y\in {\mathfrak{X}} \\ d(y, x)\leq r}}|K_{\mathfrak{X}}(x, y; a)|^2 \\
\leq \sum_{x\in D_n}\int |a(x, \omega, s)|^2 d\nu_x(\omega)d\mu(s) + q^r \sum_{\substack{x\in D_n \\ \rho(x)\leq r}}\int |a(x, \omega, s)|^2 d\nu_x(\omega)d\mu(s) 
\end{multline*}
\end{lemma}

\begin{proof}
By definition we have
\begin{align*}
\norm{\Op_{G_n}(a)}^2_{HS} &= \sum_{x,y \in D_n}  |K_{G_n}(x, y; a)|^2 \\
  &= \sum_{x,y \in D_n} \left| \sum_{\gamma\in \Gamma_n}K_{\mathfrak{X}}(x,\gamma\cdot y;a) \chi\left(\frac{d(x,\gamma\cdot y)}{r}\right) \right|^2.
\end{align*}
We split the sum into two parts, whether $\rho(x) \geq r$ or not. If $\rho(x) \geq r$, then the sum over $\gamma \in \Gamma_n$ is reduced to only one term, thanks to the cut-off function. If $\rho(x) \leq r$, then there are at most $q^r$ terms in the sum over $\gamma \in \Gamma_n$, and we can use Cauchy-Schwarz inequality to bound it as follows 
\begin{align*}
\norm{\Op_{G_n}(a)}^2_{HS} 
  &\leq \sum_{\substack{x\in D_n \\ \rho(x)\geq r}}\sum_{\substack{y\in {\mathfrak{X}} \\ d(y, x)\leq r}}|K_{\mathfrak{X}}(x, y; a)|^2 + q^r \sum_{\substack{x\in D_n \\ \rho(x)\leq r}}\sum_{\substack{y\in {\mathfrak{X}} \\ d(y, x)\leq r}}|K_{\mathfrak{X}}(x, y; a)|^2\\
  &\leq \sum_{\substack{x\in D_n \\ \rho(x)\geq r}}\sum_{y\in {\mathfrak{X}} }|K_{\mathfrak{X}}(x, y; a)|^2 
  + q^r \sum_{\substack{x\in D_n \\ \rho(x)\leq r}}\sum_{y\in {\mathfrak{X}} }|K_{\mathfrak{X}}(x, y; a)|^2.
\end{align*}
Plancherel formula for the Fourier-Helgason transform, applied to $y \mapsto K_{\mathfrak{X}}(x, y; a)$ with $x$ fixed, converts this last expression to
\begin{align*}
\norm{\Op_{G_n}(a)}^2_{HS} 
  &\leq \sum_{x\in D_n}\int |a(x, \omega, s)|^2 d\nu_x(\omega)d\mu(s)\\
  &\quad+ q^r \sum_{\substack{x\in D_n \\ \rho(x)\leq r}}\int |a(x, \omega, s)|^2 d\nu_x(\omega)d\mu(s).
\end{align*}
\end{proof}

\subsection{``Egorov''-type properties\label{s:Egorov}}
For Quantum Ergodicity on manifolds, the Egorov theorem is a statement saying that the matrix elements $\la \psi_j, \Op(a)\psi_j\ra$ remain almost invariant when transporting $a$ along the geodesic flow, when $\psi_j$ are the eigenfunctions of the Laplace-Beltrami operator $\Delta$. This is proven by showing that taking the bracket $[\Delta, \Op(a)]$ amounts to differentiating $a$ along the geodesic flow (up to some ``negligible'' error term). Here we try to perform a similar calculation.

We define a map $\sigma~: {\mathfrak{X}}\times\Omega\To {\mathfrak{X}}\times \Omega$ by $\sigma(x, \omega)=(x', \omega)$ where $x'\in {\mathfrak{X}}$ is the unique point such that $d_{\mathfrak{X}}(x, x')=1$ and $x'$ belongs to the half-geodesic $[x, \omega)$. If $[x, \omega)=(x, x_1, x_2, x_3, \ldots)$, then $\sigma(x, \omega)=(x_1, \omega)$, corresponding to the half-geodesic $[x_1, \omega)=(x_1, x_2, x_3, \ldots)$. In symbolic dynamics, $\sigma$ is the shift, and if we compare with Quantum Ergodicity on manifolds, $\sigma$ plays the role of the ``geodesic flow'' on phase space ${\mathfrak{X}}\times\Omega$. This map is not invertible, actually each point has exactly $q$ pre-images. We shall denote by $U$ the operator $a\mapsto a\circ\sigma$. For $a:{\mathfrak{X}}\times\Omega\To \C$, we define $La:{\mathfrak{X}}\times\Omega\To \C$ by
\[ La(x, \omega)=\frac1q\sum_{y\in {\mathfrak{X}}, \sigma(y, \omega)=(x, \omega)} a(y, \omega). \]
If $a$ and $b$ are compactly supported functions, we have
$$\sum_{x\in {\mathfrak{X}}}\int_{\Omega} a\circ\sigma(x, \omega) b(x, \omega)d\nu_x(\omega)=
\sum_{x\in {\mathfrak{X}}}\int_{\Omega} a(x, \omega) Lb(x, \omega)d\nu_x(\omega),$$
in other words $L$ is the adjoint of $U$ on the Hilbert space $L^2({\mathfrak{X}}\times\Omega, \sum_{x} \delta_x d\nu_x(\omega) )$. In addition, we also have $LU=I$, reflecting the fact that $U$ is an isometry of $L^2({\mathfrak{X}}\times\Omega, \sum_{x} \delta_x d\nu_x(\omega) )$.
The operators $U $ and $L $ preserve the $\Gamma_n$-invariant functions.
If $a$ and $b$ are a $\Gamma_n$-invariant functions, we still have 
\begin{equation}\label{e:adjoint}
 \sum_{x\in D_n}\int_{\Omega} a\circ\sigma(x, \omega) b(x, \omega)d\nu_x(\omega)=\sum_{x\in D_n}\int_{\Omega} a(x, \omega) Lb(x,\omega)d\nu_x(\omega).
\end{equation}
Recall that $D_n$ is a fundamental domain for the action of $\Gamma_n$ on ${\mathfrak{X}}$.

In what follows, we extend the definition of $U$ and $L$ to functions on ${\mathfrak{X}}\times \Omega\times \IR$, by a trivial action on the last component. The crucial bracket calculation is the following~:

\begin{prop}\label{p:egorov} If $a\in S({\mathfrak{X}})$ is $\Gamma_n$-invariant, then
$$[\Delta, \Op_{G_n}(a)]= \Op_{G_n}(c) + R$$
where $c \in S({\mathfrak{X}})$ is a symbol given by
\[c(x,\omega,s) = \frac{q^{1/2}}{q+1} \left( q^{is} (a\circ\sigma -a)(x,\omega,s) + q^{-is} (La-a)(x,\omega,s) \right), \]
and $R$ is a remainder such that
\begin{multline*}
 \norm{R}^2_{HS}\leq O\left(\frac{1}{r^2}\right) \Bigg( \sum_{x\in D_n}\int |a(x, \omega, s)|^2 d\nu_x(\omega)d\mu(s) \\
 + q^{r+2} \sum_{\substack{x\in D_n \\ \rho(x)\leq r+2}}\int |a(x, \omega, s)|^2 d\nu_x(\omega)d\mu(s) \Bigg).
\end{multline*}
Since $\la \psi_j, [\Delta, \Op_{G_n}(a)] \psi_j\ra=0$ for every laplacian eigenfunction $\psi_j$, this implies
\begin{multline}\label{e:egorov}
 \sum_j |\la \psi_j, \Op_{G_n}(c)\psi_j\ra|^2 
  \leq O\left(\frac{1}{r^2}\right) \Bigg( \sum_{x\in D_n}\int |a(x, \omega, s)|^2 d\nu_x(\omega)d\mu(s) \\
  + q^{r+2}  \sum_{\substack{x\in D_n\\ \rho(x)\leq r+2}}\int |a(x, \omega, s)|^2 d\nu_x(\omega)d\mu(s) \Bigg).
\end{multline}
\end{prop}

\begin{remark}
  An exact analogue of the usual Egorov theorem on manifolds would require an estimate of $\norm{\Op(c)}_{L^2}$ and show that it is $0$ up to a vanishing remainder term. Here, due to our use of the Hilbert-Schmidt norm, we will only show this for the average
  \[ \frac{1}{N(I_n,G_n)} \sum_{s^{(n)}_j\in I_n} |\la \psi^{(n)}_j, \Op_{G_n}(c)\psi^{(n)}_j\ra|^2, \]
  which is sufficient to prove our theorem.
 \end{remark}

\begin{proof}
Let us denote by $K_{G_n}(x,y; [.])$ the kernel of $[\Delta, \Op_{G_n}(a)]$. We know from \cite{LM} that $K_{\mathfrak{X}}(x,y;c)$ is the kernel of $[\Delta, \Op_{\mathfrak{X}}(a)]$. We are interested in the difference $K_{G_n}(x,y; [.]) -K_{G_n}(x,y;c)$, which is the kernel of the operator $R$. 

We have
\begin{align}
 K_{G_n}(x,y; [.]) &= \frac{1}{q+1} \left( \sum_{d(x,z)=1} K_{G_n}(z,y;a) - \sum_{d(z,y)=1} K_{G_n}(x,z;a) \right)\nonumber \\
 &= \frac{1}{q+1} \sum_{\gamma \in \Gamma_n} K_{\mathfrak{X}}(x,\gamma\cdot y),\label{e:commutatorG}
\end{align}
where
\begin{equation}\label{e:commutatorX}
  K_{\mathfrak{X}}(x,y) :=  \sum_{d(x,z)=1} K_{\mathfrak{X}}(z,y;a) \chi\left( \frac{d(z, y)}{r} \right) - \sum_{d(z, y)=1} K_{\mathfrak{X}}(x,z;a) \chi\left( \frac{d(x,z)}{r} \right)
\end{equation}
Because of the cut-off functions, the sum \eqref{e:commutatorG} only runs on those $\gamma \in \Gamma_n$ for which $d(x,\gamma\cdot y) \leq r+1$; and in \eqref{e:commutatorX} we have $K_{\mathfrak{X}}(x,y) = K_{\mathfrak{X}}(x,y) \bbbone_{ \{d(x,y) \leq r+1\} }$.

In the first sum of the right-hand side of equality \eqref{e:commutatorX}, $d(z,y) = d(x,y) \pm 1$, because $x$ and $z$ are neighbours. In the second sum $d(x,z) = d(x,y) \pm 1$, because $z$ and $y$ are neighbours. Since $\chi$ is a smooth function, both $\chi\left( \frac{d(z,y)}{r} \right)$ and $ \chi\left( \frac{d(x,z)}{r} \right)$ are equal to $\chi\left( \frac{d(x,y)}{r} \right) + O\left(\frac{1}{r}\right)$, and we have

\begin{align*}
 K_{{\mathfrak{X}}}(x,y) &= \left( \sum_{d(x,z)=1} K_{\mathfrak{X}}(z,y;a) - \sum_{d(z,y)=1} K_{\mathfrak{X}}(x,z;a) \right) \chi\left( \frac{d(x,y)}{r} \right) \\
 & \quad + O\left(\frac{1}{r}\right) \left( \sum_{d(x,z)=1} |K_{\mathfrak{X}}(z,y;a) |+ \sum_{d(z,y)=1}| K_{\mathfrak{X}}(x,z;a)| \right) \bbbone_{ \{d(x,y) \leq r+1\} } \\
 &= (q+1)K_{\mathfrak{X}}(x,y;c) \chi\left( \frac{d(x,y)}{r} \right) \\
 & \quad +O\left(\frac{1}{r}\right) \left( \sum_{d(x,z)=1} |K_{\mathfrak{X}}(z,y;a) |+ \sum_{d(z,y)=1}| K_{\mathfrak{X}}(x,z;a)| \right) \bbbone_{ \{d(x,y) \leq r+1\} }.
\end{align*}
Now if we go back to \eqref{e:commutatorG} we get $K_{G_n}(x,y; [.]) =K_{G_n}(x,y;c) + K_{G_n}(x, y; R)$, where $K_{G_n}(x, y; R)$ is the kernel of the operator $R$, given by
\[ K_{G_n}(x, y; R)= O\left(\frac{1}{r}\right) \sum_{\gamma \in \Gamma_n}  
 \left( \sum_{\substack{z\in {\mathfrak{X}} \\ d(x,z)=1}} |K_{\mathfrak{X}}(z,\gamma\cdot y;a) |+ \sum_{\substack{z\in {\mathfrak{X}} \\ d(z,\gamma\cdot y)=1}}| K_{\mathfrak{X}}(x,z;a)| \right) \bbbone_{ \{d(x,\gamma\cdot y) \leq r+1\}}.\]
 We estimate the Hilbert-Schmidt norm of $R$ by first writing
\begin{multline*}
 \sum_{x \in D_n} \sum_{y \in D_n} \left|\sum_{\gamma\in\Gamma_n} \left( \sum_{\substack{z\in {\mathfrak{X}}\\ d_{\mathfrak{X}}(x,z)=1}} |K_{\mathfrak{X}}(z,\gamma\cdot y;a) |+ \sum_{\substack{z\in {\mathfrak{X}}\\ d_{\mathfrak{X}}(z,\gamma\cdot y)=1}}| K_{\mathfrak{X}}(x,z;a)| \right) \bbbone_{ \{d_{\mathfrak{X}}(x,\gamma\cdot y) \leq r+1\}} \right|^2 \\
 =
 \sum_{x \in D_n} \sum_{y \in D_n} \left|\sum_{\gamma\in\Gamma_n} \sum_{\substack{z\in {\mathfrak{X}}\\ d_{\mathfrak{X}}(x,z)=1}} |K_{\mathfrak{X}}(z,\gamma\cdot y;a) |\bbbone_{ \{d_{\mathfrak{X}}(z,\gamma\cdot y) \leq r+2\}}+ \sum_{\substack{z\in {\mathfrak{X}}\\ d_{\mathfrak{X}}(z,\gamma\cdot y)=1}}| K_{\mathfrak{X}}(x,z;a)|  \bbbone_{ \{d_{\mathfrak{X}}(x,\gamma\cdot y) \leq r+1\}} \right|^2.
 \end{multline*}
 We then use Cauchy-Schwarz and reason along the same lines as in lemma \ref{l:HSnorm}, to bound the former expression
 by
 \begin{align*}
&\leq{2(q+1)} \Bigg( \sum_{\substack{z\in D_n \\ \rho(z)\geq r+2}} \sum_{y \in {\mathfrak{X}}}\sum_{\substack{x \in D_n \\ d_{G_n}(z, x)=1}}| K_{\mathfrak{X}}(z,y;a)| ^2  + \sum_{\substack{x \in D_n \\ \rho(x)\geq r+1}} \sum_{y \in {\mathfrak{X}}}\sum_{\substack{z \in {\mathfrak{X}}\\ d_{\mathfrak{X}}(z, y)=1}}| K_{\mathfrak{X}}(x,z;a)| ^2 \\
&+q^{r+2}\sum_{\substack{z\in D_n\\ \rho(z)\leq r+2}} \sum_{y \in {\mathfrak{X}}}\sum_{\substack{x\in D_n\\ d_{G_n}(z, x)=1}}| K_{\mathfrak{X}}(z,y;a)| ^2   +  q^{r+1} \sum_{\substack{x \in D_n\\ \rho(x)\leq r+1}} \sum_{y \in {\mathfrak{X}}}\sum_{\substack{z\in {\mathfrak{X}}\\ d_{\mathfrak{X}}(z, y)=1}}| K_{\mathfrak{X}}(x,z;a)| ^2 \Bigg)\\
&\leq {2(q+1)^2} \Bigg( \sum_{\substack{z \in D_n\\ \rho(z)\geq r+2}} \sum_{y \in {\mathfrak{X}}} |K_{\mathfrak{X}}(z,y;a)|^2  + \sum_{\substack{x \in D_n\\ \rho(x)\geq r+1}}  \sum_{z \in {\mathfrak{X}}}| K_{\mathfrak{X}}(x,z;a)| ^2 \\
&\quad + q^{r+2}\sum_{\substack{z \in D_n\\ \rho(z)\leq r+2}} \sum_{y \in {\mathfrak{X}}} |K_{\mathfrak{X}}(z,y;a)|^2  +  q^{r+1} \sum_{\substack{x \in D_n\\ \rho(x)\leq r+1}} \sum_{z \in {\mathfrak{X}}} |K_{\mathfrak{X}}(x,z;a)|^2 \Bigg)
\end{align*}
\begin{align*}
&\leq 4(q+1)^2\sum_{x\in D_n}\int |a(x, \omega, s)|^2 d\nu_x(\omega)d\mu(s) + q^{r+2} \sum_{\substack{x\in D_n \\ \rho(x)\leq r+2}}\int |a(x, \omega, s)|^2 d\nu_x(\omega)d\mu(s).
\end{align*}
Finally we have
\begin{multline*}
 \norm{R}^2_{HS}\leq O\left(\frac{1}{r^2}\right) \Bigg( \sum_{x\in D_n}\int |a(x, \omega, s)|^2 d\nu_x(\omega)d\mu(s) \\
 + q^{r+2} \sum_{\substack{x\in D_n\\ \rho(x)\leq r+2}}\int |a(x, \omega, s)|^2 d\nu_x(\omega)d\mu(s) \Bigg).
\end{multline*}
\end{proof}

In what follows, Proposition \ref{p:egorov} will be translated into an invariance property of the type 
 \[ \frac{1}{N(I_n,G_n)} \sum_{s_j\in I_n} |\la \psi_j, \Op_{G_n}(a)\psi_j\ra|^2 \sim \frac{1}{N(I_n,G_n)} \sum_{s_j\in I_n} |\la \psi_j, \Op_{G_n}(Ta)\psi_j\ra|^2 \]
 for some operator $T$. A key idea is then to take advantage of the spectral properties of $T$ and its iterates $T^k$ for $k \in \N$. In the special case where $s_0=\tau/2$ we can take $T=U^2$ (Corollary \ref{c:egorovtau2}), which makes this spectral value special. For general values of $s_0$, $T$ is a linear combination with complex coefficients of the non-commuting operators $L$ and $U$, and its spectral properties are not so nice. The aim of the successive operations done in Corollaries \ref{c:mainegorov} to \ref{c:finalegorov} is to replace $Ta$ with $Ua$ up to some error term. We first replace $a$ with $q^{is}a\circ\sigma$ in \eqref{e:egorov} to obtain

\begin{corollary}  \label{c:mainegorov}
\begin{align*}\sum_j |\la \psi_j, &\Op\left( (U-I)(q^{2is}U-I)a \right)\psi_j\ra|^2 \\
  &\leq O\left(\frac1{r^2}\right) \sum_{x\in D_n}\int |a (x, \omega, s)|^2 d\nu_x(\omega)d\mu(s)\\
  &\quad + O({q^r})|\{x\in D_n, \rho(x)\leq r+2\}| \int \norm{a(\cdot,\cdot,s)}^2_\infty d\mu(s)
\end{align*}
where $U a=a\circ\sigma$
\end{corollary}
\begin{proof}
 Recall that symbol $c$ of proposition \ref{p:egorov} is given by
 $$ c = \frac{q^{1/2}}{q+1} (q^{is} (Ua -a) + q^{-is} (La-a)). $$
 If we replace $a$ with $q^{is}Ua$ we have
 $$ c = \frac{q^{1/2}}{q+1} (q^{2is} (U^2 a - Ua) + (a - Ua)) = \frac{q^{1/2}}{q+1}(U-I)(q^{2is}U - I)a.$$
 It follows from \eqref{e:egorov} that 
 \begin{align*}
  \sum_j |\la \psi_j, &\Op\left( (U-I)(q^{2is}U-I)a \right)\psi_j\ra|^2 \\
  &\leq O\left(\frac1{r^2}\right) \sum_{x\in D_n}\int |U a (x, \omega, s)|^2 d\nu_x(\omega)d\mu(s)\\
  &\quad + O\left( \frac{q^r}{r^2}\right) \sum_{x\in D_n, \rho(x)\leq r+2}\int |U a (x, \omega, s)|^2 d\nu_x(\omega)d\mu(s) \\
  &\leq O\left(\frac1{r^2}\right) \sum_{x\in D_n}\int |a (x, \omega, s)|^2 d\nu_x(\omega)d\mu(s)\\
  &\quad +O\left( \frac{q^r}{r^2}\right) |\{x\in D_n, \rho(x)\leq r+2\}| \int \norm{a(\cdot,\cdot,s)}^2_\infty d\mu(s),
 \end{align*}
where we used the fact that $U$ preserves the $L^2$ and $L^\infty$ norms.
\end{proof}

The idea is then to invert $(q^{2is}U-I)$. As the series $\sum_k q^{2iks}U^k$ is a formal inverse to $(q^{2is}U-I)$, we apply Corollary \ref{c:mainegorov} to $a=\sum_{k=0}^{N-1} q^{2iks}U^k b:= b_{N-1}$, where $b \in S(\mathfrak{X})$ and $N$ is an arbitrary integer. We obtain
\begin{corollary}  
\begin{align*}
 &\sum_j |\la \psi_j, \Op\left( (U-I)(I-q^{2iNs}U^N  ) b\right)\psi_j\ra|^2 \\
 & \quad \leq O\left( \frac{N^2}{r^2}\right) \sum_{x\in D_n}\int |b(x, \omega, s)|^2 d\nu_x(\omega)d\mu(s)\\
 & \qquad + O\left(\frac{N^2 q^r}{r^2} \right)|\{x\in D_n, \rho(x)\leq r+2\}| \int \norm{b(\cdot,\cdot,s)}^2_\infty d\mu(s)
\end{align*}
 \end{corollary}
\begin{proof}
We apply Corollary \ref{c:mainegorov} to $a=\sum_{k=0}^{N-1} q^{2iks}U^k b:= b_{N-1}$ and use the identity
 $$(q^{2is}U - I)b_{N-1} = (q^{2iNs}U^N - I)b$$
 combined with the fact that $U$ preserves the $L^2$ and $L^\infty$ norms.
\end{proof}

If we apply Corollary \ref{c:mainegorov} to $a=\frac{1}N\sum_{k=0}^{N-1}b_k $, we obtain
\begin{corollary} \label{c:secondegorov}
\begin{multline*}\sum_j \left|\la \psi_j, \Op\left( Ub-b-q^{2is}U(U-I)b^{(s, N)}\right)\psi_j\ra\right|^2\\
\leq
  O\left( \frac{N^2}{r^2}\right) \sum_{x\in D_n}\int |b(x, \omega, s)|^2 d\nu_x(\omega)d\mu(s) \\
  +O\left(\frac{N^2 q^r}{r^2}\right) |\{x\in D_n, \rho(x)\leq r+2\}|\int \norm{b(\cdot,\cdot,s)}^2_\infty d\mu(s)
\end{multline*}
where $b^{(s, N)}=\frac1N\sum_{k=0}^{N-1} q^{2is k}U^k b$.
\end{corollary}
Note that the ``remainder term'' $q^{2is}U(U-I)b^{(s, N)}$ is not small in the symbol norm~: in Section \ref{s:general}, the (EXP) assumption will be used to show that it is small in the $L^2$-norm. This is a major difference with the Egorov theorem on manifolds, where no ergodicity assumption is needed.
\begin{proof}
 We know from the proof of the previous corollary that
 $$(I - q^{2is}U)b_{k} = (I - q^{2i(k+1)s}U^{k+1})b.$$
 It follows that
$(I - q^{2is}U)\frac{1}{N}\sum_{k=0}^{N-1} b_{k}  =  b - q^{2is} U b^{(s,N)},$
 and $$ (U-I)(I - q^{2is}U)\frac{1}{N}\sum_{k=0}^{N-1} b_{k} = Ub - b - q^{2is}U(U-I)b^{(s,N)}. $$
\end{proof}

Combining with the Hilbert-Schmidt estimate, Lemma \ref{l:HSnorm}, we get
\begin{corollary}  \label{c:finalegorov}
\begin{align*}\sum_j \left|\la \psi_j, \Op\left( Ub-b \right)\psi_j\ra\right|^2 
 &\leq O\left( \frac{N}{r^2} \right)\sum_{x\in D_n}\int |b(x, \omega, s)|^2 d\nu_x(\omega)d\mu(s)\\
 &\quad + 2 \sum_{x\in D_n}\int |b^{(s, N)}(x, \omega, s)|^2 d\nu_x(\omega)d\mu(s) \\
 &\quad + O(N^2 q^r) |\{x\in D_n, \rho(x)\leq r+2\}| \int \norm{b(\cdot,\cdot,s)}^2_\infty d\mu(s)
\end{align*}
\end{corollary}
\begin{proof}
We write
\begin{align*}
 \sum_j \left|\la \psi_j, \Op\left(Ub-b\right)\psi_j\ra \right|^2
 & \leq 2\sum_j \left|\la \psi_j, \Op\left( Ub-b-q^{2is}U(U-I)b^{(s, N)}\right)\psi_j\ra\right|^2 \\
 &\quad + 2\sum_j  \left| \la \psi_j, \Op\left( q^{2is}U(I-U)b^{(s, N)}\right)\psi_j\ra \right|^2 
\end{align*}
The first term on the right-hand side is estimated by Corollary \ref{c:secondegorov}. We estimate the last term thanks to Lemma \ref{l:HSnorm}, and we use the fact that $U$ preserves the $L^2$ norm by $U$~:
\begin{align*}
&\sum_j  \left| \la \psi_j, \Op\left( q^{2is}U(I-U)b^{(s, N)}\right)\psi_j\ra \right|^2 \\
 &\quad \leq  \sum_{x\in D_n}\int |(I-U)b^{(s, N)}(x, \omega, s)|^2 d\nu_x(\omega)d\mu(s) \\
 &\qquad + q^r \sum_{x\in D_n, \rho(x)\leq r}\int |(I-U)b^{(s, N)}(x, \omega, s)|^2 d\nu_x(\omega)d\mu(s) \\
 &\quad \leq    \sum_{x\in D_n}\int |b^{(s, N)}(x, \omega, s)|^2 d\nu_x(\omega)d\mu(s) \\
 &\qquad +  4N^2 q^r |\{x\in D_n, \rho(x)\leq r\}| \int \norm{b(\cdot,\cdot,s)}^2_\infty d\mu(s). \\
\end{align*}
\end{proof}

As already mentioned, the value $s_0=\tau/2$ is special and the previous corollaries may be replaced by the following, simpler one. In case the support of $a$ shrinks around $s_0=\tau/2$, this is a closer analogue of the Egorov theorem on manifolds in the sense that no ergodicity or expanding assumption is needed to show that the remainder term goes to $0$.\begin{corollary} \label{c:egorovtau2}We have
\begin{align*}
 \sum_j |\la \psi_j, \Op(a)-  \Op(a\circ\sigma^2 )\psi_j\ra|^2
  &\leq O\left(\frac{1}{r^2}\right) \sum_{x\in D_n}\int |a(x, \omega, s)|^2 d\nu_x(\omega)d\mu(s) \\
  &\quad + O\left( q^{r+2} \right)|\{x\in D_n, \rho(x)\leq r+2 \}| \int \norm{a(\cdot,\cdot,s)}^2_\infty d\mu(s) \\
  &\quad + n \int O(|s-s_0|^2) \norm{a(\cdot,\cdot,s)}^2_\infty d\mu(s)
\end{align*}
\end{corollary}

\begin{proof}
 We replace the symbol $a$ in proposition \ref{p:egorov} with $q^{is}a\circ\sigma$. As $L(a\circ\sigma) = a$, the symbol $b$ becomes $b = a - T a$, where
 $$T a(x,\omega,s)=-q^{2is} a\circ \sigma^2(x,\omega,s) +2 q^{is} \cos (s\ln q) a\circ\sigma(x,\omega,s)$$
 and we have
 $\sum_j |\la \psi_j, \Op(a)- \Op(T a)\psi_j\ra|^2 \leq \norm{R}^2_{HS},$
 where
 \begin{align*}
 \norm{R}^2_{HS} &\leq O\left(\frac{1}{r^2}\right) \sum_{x\in D_n}\int |a\circ\sigma(x, \omega, s)|^2 d\nu_x(\omega)d\mu(s)\\
 &\quad + O\left( q^{r+1} \right) \sum_{x\in D_n, \rho(x)\leq r+1}\int |a\circ\sigma(x, \omega, s)|^2 d\nu_x(\omega)d\mu(s)\\
 &\leq  O\left(\frac{1}{r^2}\right)\sum_{x\in D_n}\int |a(x, \omega, s)|^2 d\nu_x(\omega)d\mu(s) \\
 &\quad + O\left( q^{r+1} \right) |\{x\in D_n, \rho(x)\leq r+1 \}|  \int \norm{a(\cdot,\cdot,s)}^2_\infty d\mu(s) .
 \end{align*}
 
 Now $T a = a\circ\sigma^2 - c$, with
 $$c(x,\omega,s) = (1+q^{2is}) a\circ \sigma^2(x,\omega,s) +2 q^{is} \cos (s \ln q) a\circ\sigma(x,\omega,s)$$
 and we can write
  \begin{align*}
  \sum_j |\la \psi_j, \Op(a)- \Op(a\circ\sigma^2)\psi_j\ra|^2 
  &\leq 2 \sum_j |\la \psi_j, \Op(c)\psi_j\ra|^2 + 2 \norm{R}^2_{HS} \\
  &\leq 2 \norm{\Op(c)}_{HS}^2 + 2 \norm{R}^2_{HS}.
 \end{align*}
 Recalling that $s_0 = \frac{\pi}{2 \ln q}$, we have
 \begin{align*}
  \norm{\Op(c)}_{HS}^2 &\leq \sum_{x\in D_n}\int |(1+q^{2is}) a\circ \sigma^2+2 q^{is} \cos (s \ln q) a\circ\sigma)|^2 d\nu_x(\omega)d\mu(s)\\ 
  & \quad + q^r \sum_{x\in D_n,  \rho(x)\leq r}\int |(1+q^{2is}) a\circ \sigma^2+2 q^{is} \cos (s \ln q) a\circ\sigma)|^2 d\nu_x(\omega)d\mu(s) \\
  &\leq \sum_{x\in D_n}\int O(|s-s_0|^2)(|a\circ \sigma^2|+ |a\circ\sigma|)^2(x,\omega,s) d\nu_x(\omega)d\mu(s)\\ 
  & \quad + q^r \sum_{x\in D_n,  \rho(x)\leq r}\int O(|s-s_0|^2) (|a\circ \sigma^2|+ |a\circ\sigma|)^2(x,\omega,s) d\nu_x(\omega)d\mu(s) \\
  &\leq n \int O(|s-s_0|^2) \sup_{x,\omega} |a(x,\omega,s)|^2 d\mu(s)\\ 
  & \quad + q^r |\{x\in D_n,  \rho(x)\leq r \}| \int  O(|s-s_0|^2) \sup_{x,\omega} |a(x,\omega,s)|^2 d\mu(s)
 \end{align*}
\end{proof}

\subsection{Two more formulas about $\Op_{G_n}(\chi_n)$}

\begin{lemma}\label{l:chin}$$\Op_{G_n}(\chi_n)\psi^{(n)}_j=\lambda^{(n)}_j \psi_j^{(n)}$$
with $$\lambda^{(n)}_j=\chi_n(s_j) + r^3 O\left(\frac{1}{(r\delta_n)^\infty}\right)$$
if $s_j\in [0, \tau]$ (tempered eigenfunctions).
\end{lemma}

\begin{proof}
 First note that $\psi_j^{(n)}$ is associated to a $\Gamma_n$-invariant eigenfunction of the laplacian on the tree ${\mathfrak{X}}$, that we will still denote by $\psi_j^{(n)}$. We have on the tree
 \begin{align*}
  \Op_{G_n}(\chi_n) \psi_j^{(n)}(x) 
  &= \sum_{y\in {\mathfrak{X}}} K_{{\mathfrak{X}}}(x,y;\chi_n) \chi\left(\frac{d(x,y)}{r}\right) \psi_j^{(n)}(y)
 \end{align*}
 and $K_{{\mathfrak{X}}}(x,y;\chi_n) \chi\left(\frac{d(x,y)}{r}\right)$ depends only on $d(x,y)$ because $\chi_n$ does not depend on $(x,\omega)$.
 We thus have 
 $$\Op_{G_n}(\chi_n) \psi_j^{(n)}(x)  = f(s_j) \psi_j^{(n)}(x)$$
 where $f(s_j)$ is given by the spherical transform of the kernel
 $$f(s_j) = \sum_{y\in {\mathfrak{X}}} K_{{\mathfrak{X}}}(x,y;\chi_n) \chi\left(\frac{d(x,y)}{r}\right) \phi_{s_j}(d(x,y))$$
 and $\phi_{s_j}$ is the spherical function associated to $s_j$ defined in \eqref{e:spherical}. Now
 \begin{align*}
  f(s_j) &= \sum_{y\in {\mathfrak{X}}} K_{{\mathfrak{X}}}(x,y;\chi_n) \phi_{s_j}(d(x,y)) \\
  &\quad -  \sum_{y\in {\mathfrak{X}}} K_{{\mathfrak{X}}}(x,y;\chi_n) \phi_{s_j}(d(x,y)) \left(1 - \chi\left(\frac{d(x,y)}{r}\right)\right) \\
  &= \chi_n(s_j) \\
  &\quad -  \sum_{y\in {\mathfrak{X}}} K_{{\mathfrak{X}}}(x,y;\chi_n) \phi_{s_j}(d(x,y)) \left(1 - \chi\left(\frac{d(x,y)}{r}\right)\right).
 \end{align*}

 Because $\chi_n \in S({\mathfrak{X}})$, according to the rapid decay property of the kernel of pseudodifferential operators \eqref{e:decay}, we have
 $$ |K_{{\mathfrak{X}}}(x,y;\chi_n)| \leq C \left(\sum_{k=0}^{M+1} \norm{\partial_s^k \chi_n}_\infty \right) \frac{q^{-\frac{d(x,y)}{2}}}{(1+d(x,y))^M} $$ 
 for every $M \in \N$. Moreover, if $s_j$ is a tempered eigenvalue, then 
 $$ |\phi_{s_j}(d(x,y))| \leq C q^{-\frac{d(x,y)}{2}}. $$
 So, using also the fact that $\norm{\partial_s^k \chi_n}_\infty \leq C_k \delta_n^{-k}$, we obtain that for every $M \in \N$
  \begin{align*}
  |f(s_j) - \chi_n(s_j)|
  &\leq \frac{C_{M+1}}{\delta_n^{M+1}}\sum_{y\in {\mathfrak{X}}} \frac{q^{-d(x,y)}}{(1+d(x,y))^M} \left(1 - \chi\left(\frac{d(x,y)}{r}\right)\right) \\
  &\leq \frac{C_{M+1}}{\delta_n^{M+1}(1+r)^{M-2}} \sum_{y\in {\mathfrak{X}}} \frac{q^{-d(x,y)}}{(1+d(x,y))^2}\\
  &= \frac{C_{M+1}}{\delta_n^{M+1}(1+r)^{M-2}} \sum_{k\in\N} \sum_{y:d(x,y)=k} \frac{q^{-k}}{(1+k)^2}\\
  &\leq \frac{C_{M+1}}{\delta_n^{M+1}(1+r)^{M-2}} \sum_{k\in\N} \frac{1}{(1+k)^2}.
 \end{align*}

 We thus have $$|f(s_j) - \chi_n(s_j)| = r^3 O\left(\frac{1}{(r\delta_n)^{M+1}}\right)$$
 for any $M$.
\end{proof}

\begin{lemma}\label{l:produit}Fix an integer $D$.
Let $a$ be such that $K_{\mathfrak{X}}(x,y;a)=0$ for $d_{\mathfrak{X}}(x, y)>D$ (in other words, $a\in S_o^D({\mathfrak{X}})$) and $\varphi=\varphi(s)$. We have
$$\Op_{G_n}(a\varphi)=\Op_{G_n}(a) \Op_{G_n}(\varphi)+R$$
where
$$\norm{R}_{HS}^2\leq  O(r^{-1})\int\varphi^2(s)d\mu(s)\left(n+|\{x\in V_n, \rho(x)\leq r+D\}| q^{r+D}\right)$$
\end{lemma}
\begin{proof} The kernel $K_{G_n}(x, z; a\varphi)$ of $\Op_{G_n}(a\varphi)$ is obtained by the periodization
$$K_{G_n}(x, z; a\varphi)=\sum_{\gamma\in\Gamma_n}K_{{\mathfrak{X}}}(x, \gamma\cdot z; a\varphi)\chi\left(\frac{d(x, \gamma\cdot z)}r\right).$$
Because $\varphi$ only depends on $s$, we note that
$$K_{{\mathfrak{X}}}(x,   z; a\varphi)=\sum_{y\in {\mathfrak{X}}}K_{{\mathfrak{X}}}(x,   y; a )K_{{\mathfrak{X}}}(y,  z ; \varphi)$$
(see \eqref{e:obvproduct}) and
\begin{eqnarray*}K_{{\mathfrak{X}}}(x,   z; a\varphi)\chi\left(\frac{d(x,  z)}r\right)&=&\sum_{y\in {\mathfrak{X}}}K_{{\mathfrak{X}}}(x,   y; a )K_{{\mathfrak{X}}}(y,  z ; \varphi)\chi\left(\frac{d(y,  z)}r\right)\\&&+
\sum_{y\in {\mathfrak{X}}}K_{{\mathfrak{X}}}(x,   y; a )K_{{\mathfrak{X}}}(y,  z ; \varphi)\left(\chi\left(\frac{d(x,  z)}r\right)-\chi\left(\frac{d(y,  z)}r\right)\right)\\
&&= \sum_{y\in {\mathfrak{X}}}K_{{\mathfrak{X}}}(x,   y; a )K_{{\mathfrak{X}}}(y,  z ; \varphi)\chi\left(\frac{d(y,  z)}r\right)\\&&+O(r^{-1})
\sum_{y\in {\mathfrak{X}}}|K_{{\mathfrak{X}}}(x,   y; a )| |K_{{\mathfrak{X}}}(y,  z ; \varphi)| \bbbone_{d_{\mathfrak{X}}(x, z)\leq r +D}
\end{eqnarray*}
After $\Gamma_n$-periodization, we note that
$\sum_{\gamma\in\Gamma_n}\sum_{y\in {\mathfrak{X}}}K_{{\mathfrak{X}}}(x,   y; a )K_{{\mathfrak{X}}}(y, \gamma\cdot z ; \varphi)\chi\left(\frac{d(y,  \gamma\cdot z)}r\right)$ is the kernel of $\Op_{G_n}(a) \Op_{G_n}(\varphi)$ (as soon as $D<r$). Using Cauchy-Schwarz and the fact that $K_{{\mathfrak{X}}}(x,   y; a )$ is supported near the diagonal, the Hilbert-Schmidt norm of the operator with kernel
$$ \sum_{\gamma\in\Gamma_n}\sum_{y\in {\mathfrak{X}}}|K_{{\mathfrak{X}}}(x,   y; a )| |K_{{\mathfrak{X}}}(y,  \gamma\cdot z ; \varphi)| \bbbone_{d_{\mathfrak{X}}(x,\gamma\cdot z)\leq r +D}$$
on $L^2(G_n)$ can be bounded by $$q^{2D}\sup_{x, y}|K(x, y;a)|^2 \int\varphi^2(s)d\mu(s)\left(n+|\{x\in V_n, \rho(x)\leq r+D\}| q^{r+D}\right).$$

\end{proof}

\section{The proof for arbitrary $s_0$}\label{s:general}
\subsection{Upper bound on the variance}  

Fix an integer $T>0$. Let $\chi$ be a smooth cut-off function supported in $[-1, 1]$ and taking the constant value $1$ on $[-1/2, 1/2]$. We write
$$\chi_n(s)=\chi\left(\frac{s-s_0}{2\delta_n}\right)$$
so that $\chi_n\equiv 1$ on $I_n$. We use the pseudodifferential calculus and the notation defined in Section \ref{s:PDC}, taking the cut-off parameter $r$ equal to $r_n$ (from condition (EIIR), as explained before the statement of theorem \ref{t:main}).

To simplify the notation, we will write $\psi_j = \psi_j^{(n)}$, $s_j = s_j^{(n)}$, and $a = a_n$. Thanks to Lemmas \ref{l:chin}
and to the ``Egorov property'' Corollary \ref{c:finalegorov}, we have\footnote{To prove the extended Theorem \ref{t:gen}, we also need Lemma \ref{l:produit}.}
\begin{eqnarray*}
\sum_{j=0}^n \chi_n(s_j)^2 \left| \la \psi_j, a\psi_j\ra \right|^2
  &=&  \sum_{j=0}^n \left| \la \psi_j, \Op_{G_n}(a\chi_n)\psi_j\ra\right|^2 +n r^3 O\left(\frac{1}{(r\delta_n)^\infty} \right)\\
  &=& \sum_{j=0}^n \left| \la \psi_j, \Op_{G_n}(a^T\chi_n)\psi_j\ra\right|^2+n r^3 O\left(\frac{1}{(r\delta_n)^\infty} \right)\\
  && +O\left(T^2 \frac{N^2}{r^2} \right)\sum_{x\in D_n}\int |a(x, \omega)|^2 d\nu_x(\omega) \chi_n(s)^2 d\mu(s)\\
  && + O(T^2N^2 q^{r}) |\{x\in D_n, \rho(x)\leq r+2\}| \int \chi_n(s)^2 d\mu(s)\\
  && + O(T^2) \sum_{x\in D_n}\int |a^{(s, N)}(x, \omega)|^2 d\nu_x(\omega) \chi_n(s)^2 d\mu(s)
\end{eqnarray*}
where 
$$a^T:=\frac1T\sum_{k=0}^{T-1} a\circ \sigma^{k} \quad\text{ and }\quad a^{(s,N)} = \frac1N\sum_{k=0}^{N-1} q^{2is k} a\circ\sigma^k.$$

We use Lemma \ref{l:HSnorm} to write
\begin{eqnarray*} \sum_{j=0}^n \left| \la \psi_j, \Op_{G_n}(a^T\chi_n)\psi_j\ra\right|^2&\leq &
{\norm{\Op_{G_n}(a^T\chi_n)}^2}_{HS}\\
&\leq &\sum_{x\in D_n}\int |a^T(x, \omega)|^2 d\nu_x(\omega) \chi_n(s)^2 d\mu(s)\\&& + q^r \sum_{x\in D_n, \rho(x)\leq r}\int |a^T(x, \omega)|^2 d\nu_x(\omega) \int \chi_n(s)^2d\mu(s).
\end{eqnarray*}
  We have seen in Section \ref{s:ergodicity} that $\sum_{x\in D_n}\int |a^T(x, \omega)|^2 d\nu_x(\omega) =O\left(\frac{n}{T \beta}\right)$.
 The same proof shows that $\sum_{x\in D_n}\int |a^{(s, N)}(x, \omega)|^2 d\nu_x(\omega)=O\left(\frac{n}{N\beta}\right)$.
 A major difference here with the usual Quantum Ergodicity (and with the special proof of \S \ref{s:altern}) is that condition (EXP) is used already to show that the ``remainder term'' $a^{(s, N)}$ of the Egorov theorem is small in the $L^2$-norm.
 
 \begin{remark} For $s$ staying away from $\tau/2$, a slightly more careful proof would show that we only need to assume here that
 the spectrum of $A$ is contained in $\{1\}\cup [-1, 1-\beta]$. Hence our Remark \ref{r:weakexp}.
 \end{remark}
 Recall also, from the Kesten-McKay law (Section \ref{s:weyl}, Corollary \ref{c:weyl}) that
 $$ \frac{n}{N(I_n, G_n)}\int\chi_n(s)^2 d\mu(s)=O(1).$$
We obtain finally
 \begin{eqnarray*}
 \frac1{N(I_n, G_n)} \sum_{s_j\in I_n} \left| \la \psi_j, a\psi_j\ra\right|^2&= &
 r^3 O\left(\frac{1}{\delta_n(r\delta_n)^\infty} \right)+O\left(\frac{N^2T^2}{r^2}\right)\\
  &&+O(N^2T^2q^{r}\alpha_n) +O\left(\frac{T^2}{N\beta}\right)+O\left(\frac1{T\beta}\right),
 \end{eqnarray*}
 and if we choose the sequences $r = r_n$ and $\delta_n$ as explained in section \ref{s:weyl}, 
 $$ \limsup_{n \to \infty} \frac1{N(I_n, G_n)} \sum_{s_j\in I_n} \left| \la \psi_j, a\psi_j\ra\right|^2=
 O\left(\frac{T^2}{N\beta}\right)+O\left(\frac1{T\beta}\right). $$
 As the left-hand side of the equality does not depend on $T$ and $N$, we take the limit $N \To \infty$ and then $T \To \infty$ to get
 $$ \lim_{n \to \infty} \frac1{N(I_n, G_n)} \sum_{s_j\in I_n} \left| \la \psi_j, a\psi_j\ra\right|^2=0. $$

 \section{Kesten-McKay law for sequences of graphs satisfying (EIIR) \label{s:weyl}}
 In this section we give an alternative proof of the Kesten-McKay law \cite{Kes59, McK81}, which gives the spectral density for large regular graphs satisfying (EIIR) and is analogous to the Weyl law for the spectral density of the laplacian on Riemannian manifolds. Note that we consider the density of eigenvalues in intervals that are allowed to shrink as $n\To +\infty$.

In the definition of $\Op_{G_n}$, we take $r=r_n$ such that $r_n + 2 \leq R_n$ and $q^{r_n}\alpha_n \To 0$ (where $R_n$ and $\alpha_n$ are the quantities occurring in (EIIR))\footnote{We can take for example 
$ r_n = \min \left\{ R_n - 2, -(1-\epsilon)\frac{\log\alpha_n}{\log q}\right\}, $
for any $0 < \epsilon < 1$.}. We also assume that there exists an integer $M$ such that $$\frac{1}{\delta_n^{M+1} r_n^{M-3}} \To 0 .$$ If $\chi_n$ is the function defined in \eqref{e:chin} with $s_0\in (0, \tau)$, this ensures that \[ \int_{0}^\tau \chi_n(s)^2d\mu(s)\gg  r_n^3(\delta_n r_n)^{-M}.\]

\begin{theorem} Assume (EIIR). Let $\chi=\chi_n$ be a smooth function satisfying $$\norm{\partial_s^k \chi}\leq C_k \delta_n^{-k},$$ with $\delta_n r_n \To +\infty$, such that
$\frac{1}{\delta_n^{M+1} r_n^{M-3}} \To 0 $ for some $M$.

Then we have
\[\frac1n\sum_{j=1}^n \chi_n(s_j)^2\sim \int_{0}^\tau \chi_n(s)^2d\mu(s)\]
when $n \To +\infty$.
\end{theorem}

\begin{proof}
\begin{align*}
\sum_{j=1}^n \chi_n(s_j)^2 
&= \Tr\left(\Op_{G_n}(\chi_n)^2\right) + O(nr_n^3(\delta_n r_n)^{-\infty})\\
&=\sum_{x\in D_n} \sum_{y\in D_n}|K_{G_n}(x, y;\chi_n)|^2 + O(nr_n^3 (\delta_n r_n)^{-\infty})\\
&= \sum_{\rho(x)> r_n} \sum_{y\in {\mathfrak{X}}}|K_{{\mathfrak{X}}}(x, y;\chi_n)|^2 + O(nr_n^3 (\delta_n r_n)^{-\infty})\\
&\quad +O(q^{r_n})\sum_{\rho(x)\leq r_n} \sum_{y\in {\mathfrak{X}}}|K_{{\mathfrak{X}}}(x, y;\chi_n)|^2\\
&= \sum_{x\in D_n} \sum_{y\in {\mathfrak{X}}}|K_{{\mathfrak{X}}}(x, y;\chi_n)|^2 + O(nr_n^3 (\delta_n r_n)^{-\infty}) \\ 
&\quad +\left(O(q^{r_n})-1\right)\sum_{\rho(x)\leq r_n} \sum_{y\in {\mathfrak{X}}}|K_{{\mathfrak{X}}}(x, y;\chi_n)|^2\\
&= n\int_{0}^\tau \chi_n(s)^2d\mu(s) + O(nr_n^3 (\delta_n r_n)^{-\infty}) \\
&\quad + \left(O(q^{r_n})-1\right)|\{x\in D_n, \rho(x)\leq r_n\}|\int_{0}^\tau \chi_n(s)^2d\mu(s).
\end{align*}
Thus, we get the desired result if 
$$\int_{0}^\tau \chi_n(s)d\mu(s)\gg r_n^3(\delta_n r_n)^{-\infty}.$$
\end{proof}

\begin{corollary} \label{c:weyl}Under the assumptions of Theorem \ref{t:main}, we have
$$N(I_n, G_n)\sim n \int_{s_0 -\delta_n}^{s_0+\delta_n}d\mu(s)\sim 2n \delta_n |c(s_0)|^{-2}$$
where $|c(s_0)|^{-2}$ is the density of the Plancherel measure at $s_0$.
\end{corollary}
\section{Quantitative statement\label{s:quantitative}}

In this section we will give explicit upper bounds on the rate of convergence, first in terms of the parameters $R_n$ and $\alpha_n$ associated with the sequence of graphs $(G_n)$ in condition (EIIR), then depending only on $n$ for sequences of random graphs.
These results are certainly not optimal because some of our inequalities were written in a non optimal way.

In the general case $s_0 \in (0,\tau)$, we have
\begin{lemma}\label{l:quant2}
If $\delta_n = r_n^{-1+\epsilon}$ for some $0 < \epsilon < 1$, then we have
\[ \frac1{N(I_n, G_n)} \sum_{s_j^{(n)}\in I_n} \left| \la \psi_j^{(n)}, a_n\psi_j^{(n)}\ra\right|^2 = O \lp r_n^{-2/9} + r_n^{16/9} q^{r_n}\alpha_n \rp, \]
 where we can take $r_n = \min \left\{ R_n - 2, -(1-\epsilon')\log_q(\alpha_n) \right\},$
for any $0 < \epsilon' < 1$.
\end{lemma}
\begin{proof}
 According to the proof of section \ref{s:general}, we have
  \begin{align*}
 \frac1{N(I_n, G_n)} \sum_{s_j\in I_n} \left| \la \psi_j, a_n\psi_j\ra\right|^2
 &\leq r_n^3 O\left(\frac{1}{\delta_n(r_n\delta_n)^\infty} \right)+O\left(\frac{N^2T^2}{r_n^2}\right)\\
 &\quad +O(N^2T^2q^{r_n}\alpha_n) +O\left(\frac{T^2}{\beta N}\right)+O\left(\frac1{T\beta}\right).
 \end{align*}
 
 Take $N = r_n^{2/3}$ and $T = r_n^{2/9}$ such that $ \frac1T = \frac{T^2}N = \frac{N^2T^2}{r_n^2} = r_n^{-2/9}$. For every $M > 0$, we have $O\left(\frac{r_n^3}{\delta_n(r_n\delta_n)^\infty} \right) = O\lp r_n^{3-M}\delta_n^{-(1+M)}\rp =  O\lp r_n^{4 - (1+M)\epsilon} \rp$ and this term can be made negligible in comparison with the other terms by taking $M$ sufficiently large. Finally $N^2T^2q^{r_n}\alpha_n = r_n^{16/9}q^{r_n}\alpha_n$.
\end{proof}

\begin{remark} Here we kept the spectral gap $\beta$ fixed, but we see that this could be relaxed to $\beta\gg r_n^{-2/9}.$
\end{remark}

For sequences of random graphs, we have
\begin{lemma}
 Let $\delta > 1/2$, $\epsilon > 0$, and $\delta_n = (\log_q(n^{1-\delta}))^{1-\epsilon}$.
 If $G_n$ is chosen uniformly at random among the $(q+1)$-regular graphs with $n$ vertices, we have
 \[ \frac1{N(I_n, G_n)} \sum_{s_j^{(n)}\in I_n} \left| \la \psi_j^{(n)}, a_n\psi_j^{(n)}\ra\right|^2 = O \lp \log_q(n)^{-2/9} \rp, \]
 with overwhelming probability.
\end{lemma}

\begin{proof}
 Take $R_n$ and $\alpha_n$ as in example \ref{ex:random}. Let $1/2 < \delta < 1$, then 
$R_n = (1-\delta)\log_q(n)$ and $\alpha_n = 80 (1-\delta) \log_q(n) n^{-\delta}$. In this case, we take
\[r_n = \min \left\{ R_n - 2, -(1-\epsilon')\frac{\log\alpha_n}{\log q}\right\} = R_n - 2,\]
and apply lemma \ref{l:quant2}.
\end{proof}

\section{Proof of Theorem \ref{t:gen}}
Most steps of the proof carry over to arbitrary $a\in S_o^D({\mathfrak{X}})$. Actually, all that needs modifying is the treatment of the expression
\begin{equation}\label{e:neumann}\sum_{x\in D_n}\int \left|\frac1N\sum_{k=0}^{N-1}q^{2isk}a\circ\sigma^k\right|^2(x, \omega, s)d\nu_x(\omega)d\mu(s)\end{equation}
that is used in \S \ref{s:ergodicity} (for $s=0$) and Section \ref{s:general} (for $s$ close to $s_0$).
Equation \eqref{e:neumann} is also
\begin{multline}\label{e:neumann2}\frac1{N^2}\sum_{x\in D_n}\int \sum_{k=0}^{N-1}\sum_{j=0}^kq^{2isj}a\circ\sigma^{j} (x, \omega, s)a  (x, \omega, s)d\nu_x(\omega)d\mu(s)\\
+\frac1{N^2}\sum_{x\in D_n}\int \sum_{k=0}^{N-1}\sum_{j={1}}^{k}q^{-2is j}a\circ\sigma^{j} (x, \omega, s)a (x, \omega, s)d\nu_x(\omega)d\mu(s)
\end{multline}
In \S \ref{s:ergodicity}, the integral $$\sum_{x\in D_n}\int a\circ \sigma^{j}(x, \omega) a(x, \omega) d\nu_x(\omega)$$
was rewritten as $\sum_{x\in D_n} S_j a(x) a(x)$ using the fact that $a$ did not depend on $\omega$ -- thus establishing a link between the shift $\sigma$ and the laplacian. We need to adapt that argument to the case when $a(x, \omega)$ depends on the first $D$ coordinates of the half geodesic $[x, \omega)=(x, x_1, x_2, \ldots)$.

\subsection{Proof when $D=2$.}
When $D=2$, $a(x, \omega)=a(x, x_1)$, so that $a$ is a function on the set $B$ of directed bonds of $G=G_n$\footnote{Note that $B$ has cardinality $n(q+1)$ if $G$ has $n$ vertices and is $(q+1)$-regular.}.
We use the notation of \cite{Smi07}~: if $e$ is an element of $B$, we shall denote by $o(e)\in V_n$ its origin, $t(e)\in V_n$ its terminus, and $\hat e\in B$ the reversed bond.

One sees that
$$\sum_{x\in D_n}\int a\circ \sigma^{j}(x, \omega) a(x, \omega) d\nu_x(\omega)= 
\frac1{q+1} \sum_{e\in B} M^{\sharp j} a(e) a(e)
$$
where $M^\sharp$ is a bistochastic matrix indexed by $B$, defined by
$$M^\sharp(e, e')=\frac1q$$
if $o(e')=t(e)$ and $e'\not=\hat e$; and $M^\sharp(e, e')=0$ otherwise.
This is (up to normalization) the matrix appearing in \S 3 of \cite{Smi07}. It is the (normalized) adjacency matrix of the $q$-regular directed graph, whose vertices are the directed bonds of $G$, and where we draw an edge between two bonds if they are consecutive without allowing back-tracking. What we need is an explicit relation between the spectrum of $M^\sharp$ and the spectrum of the discrete laplacian on $G$, in other words, of the matrix $A$. The relation between the eigenvalues is formula (44) in \cite{Smi07}, but since we also need relations between the eigenfunctions, we shall be more explicit below. We did not write all the detailed calculations because they are lengthy but basic. We assume these relations must already be known but did not find any reference.

\begin{itemize}
\item[(o)] Both $A$ and $M^\sharp$ have $1$ in their spectrum, corresponding to the constant eigenfunction. The matrix $A$ has $-1$ in its spectrum iff the graph $G$ is bi-partite, in which case $M^\sharp$ also trivially has $-1$ in its spectrum.

\item[(i)] each eigenvalue $\lambda\not=\pm 1$ of $A$ gives rise to the two eigenvalues 
$$\frac{2}{(q+1)\left(\lambda\pm \sqrt{\lambda^2-\frac{4q}{(q+1)^2}}\right)}$$
of $M^\sharp$;

\item[(ii)] in addition, $M^\sharp$ admits the eigenvalue $ 1/q$ with multiplicity $b:=|E_n|-|V_n|+1$ (the rank of the fundamental group of $G$); and the eigenvalue $ -1/q$ with multiplicity $b-1 $ if $-1$ is not an eigenvalue of $A$, or $b $ if $-1$ is an eigenvalue of $A$.\footnote{Or, equivalently, if the graph is bi-partite.}
\end{itemize}

In particular, the eigenvalue $1$ of $M^\sharp$ has multiplicity $1$. The tempered spectrum of $A$ corresponds to eigenvalues of $M^\sharp$ of modulus $1/\sqrt{q}$; the untempered spectrum of $A$ contained in $[-1, 1-\beta]$ gives rise to real eigenvalues of $M^\sharp$ contained in $[-1, 1-\beta']$
with
$$1-\beta'= \frac{2}{(q+1)\left(1-\beta- \sqrt{(1-\beta)^2-\frac{4q}{(q+1)^2}}\right)}.$$

Since $M^\sharp$ is not normal, the knowledge of its spectrum is not sufficient to control the growth of $M^{\sharp k} $
in a precise manner (we need a bound that is independent of the size of the matrix, in other words, independent of $n$). Below, we describe explicitly the eigenvectors of $M^\sharp$ in terms of those of $A$; these eigenvectors do not form an orthogonal family but this is compensated by the fact that one can compute their scalar products explicitly.

\begin{itemize}
\item[(i)] an eigenfunction $\phi$ of $A$ for the eigenvalue $\lambda\not= \pm 1$ gives rise to the two eigenfunctions of $M^{\sharp}$,
$$f_1(e)=\phi(t(e))-\eps_1\phi(o(e));\qquad f_2(e)=\phi(t(e))-\eps_2\phi(o(e)),$$
where $\eps_1, \eps_2$ are the two roots of $q\eps^2-(q+1)\lambda\eps+1=0$ (in what follows we index them so that $|\eps_1|\leq |\eps_2|$).
Special attention has to be paid to the case $\lambda=\pm \frac{2\sqrt{q}}{q+1}$, for which $\eps_1=\eps_2$ (see below).
\item[(ii)] the eigenvalues $\pm 1/q$ of $M^\sharp$ correspond, respectively, to odd and even\footnote{Odd means $f(\hat e)=-f(e)$ and even means $f(\hat e)=f(e)$, for every bond $e$.} solutions of $\sum_{o(e)=x} f(e)=0$ (for every vertex $x$).
For the eigenvalue $1/q$, an explicit basis of eigenfunctions is indexed by generators of the fundamental group, $(\gamma_1, \ldots, \gamma_b)$~: every closed circuit $\gamma$ made of consecutive edges $(e_1, \ldots, e_k)$ gives rise to an odd eigenfunction
$$f_{\gamma}=\sum_{j=1}^k \delta_{e_j}- \delta_{\hat e_j}.$$
If $G$ is bi-partite, then all circuits have even length and we have an explicit
basis of even eigenfunctions for the eigenvalue $-1/q$, again indexed by generators of the fundamental group~:
$$g_{\gamma}=\sum_{j=1}^k (-1)^k (\delta_{e_j}+\delta_{\hat e_j})$$
if $\gamma$ is a closed circuit made of consecutive edges $(e_1, \ldots, e_k)$.
If $G$ is not bi-partite, there are closed circuits of odd lengths, in which case $g_{\gamma}$ is not an eigenfunction of $M^\sharp$. Nevertheless, if $\gamma, \gamma'$ are two circuits of odd lengths, $g_{\gamma}-g_{\gamma'}$  is now an eigenfunction of $M^\sharp$ for the eigenvalue $-1/q$.
\end{itemize}
The eigenfunctions of the family (ii) are automatically orthogonal to those of the family (i). 
In (i), eigenfunctions of $M^\sharp$ stemming from different eigenvalues $\lambda$ of $A$ are orthogonal; however, the two eigenfunctions $f_1, f_2$ stemming from the same $\lambda$ are not orthogonal.

%If the eigenfunction $\phi$ considered in (i) is normalized in $L^2(G)$, then $f_i$ has squared-norm
%$$(q+1)\left(1+|\eps_i|^2-2\Re e \eps_i \lambda\right),$$
%which is $(q-1)(1-\eps_i^2)$ in the untempered case (when $\eps_i$ are real numbers) and $\frac{(q+1)^2}q(1-\lambda^2)$ in the tempered case. For $i=1$ all these numbers are bounded away from $0$.

To evaluate the norm of a matrix, it is safer to work in an orthogonal basis, and thus we shall consider, instead of a pair $(f_1, f_2)$, the pair
$$f_1(e)=\phi(t(e))-\eps_1\phi(o(e)), f'_2(e)= \phi(t(e))-\mu\phi(o(e))$$
which can be checked to be orthogonal for
$$\mu=\frac{\bar \eps_1\lambda-1}{\bar\eps_1-\lambda}.$$
%The squared-norm of $f'_2$ is $(q+1)\left(1+|\mu|^2-2\Re e \mu\lambda\right)=(q+1)(|1-\mu|^2+2\Re e\mu(1-\lambda))=(q+1)(|1+\mu|^2-2\Re e\mu(1+\lambda))$, which is bounded away from $0$ if $\lambda$ stays in $[-1, 1-\beta]$.
In the plane generated by $\left(\frac{f_1}{\norm{f_1}}, \frac{f'_2}{\norm{f'_2}}\right)$, $M^\sharp$ has matrix
$$\left( \begin{array}{cc} 

1/q\eps_1 & \star \\
 0 & 1/q\eps_2 \\
 
 \end{array}  \right)$$
where $\star$ is a number that can be calculated explicitly in terms of $\eps_1, \eps_2$ and $\lambda$, and which is uniformly bounded (since the norm of $M^\sharp$, anyway, is bounded independently of $n$).

This discussion is also valid for $\lambda=\pm\frac{2\sqrt{q}}{q+1}$, a special case where $\eps_1=\eps_2=\pm\frac{1}{\sqrt{q}}$.

To summarize, the spectrum of $M^\sharp$ is contained in $[-1, 1-\beta']\cup\{1\}$ (resp. $[-1+\beta', 1-\beta']\cup\{1\}$) if the spectrum of $A$ is contained in $[-1, 1-\beta]\cup\{1\}$ (resp. $[-1+\beta, 1-\beta]\cup\{1\}$). We can find an orthonormal basis of $L^2(\C^B)$ in which $M^\sharp$ is block diagonal, each diagonal block being an upper triangular matrix of size $\leq 2$, and the non-diagonal coefficients are uniformly bounded.

This implies that the operator
$$\frac{1}{N^2}\sum_{k=0}^{N-1}\sum_{j=0}^k  q^{2isj}M^{\sharp j}=\frac{1}{N^2}(q^{2is}M^\sharp-I)^{-2}\left( q^{2is(N+1)}M^{\sharp( N+1)}-q^{2is}M^\sharp-N(q^{2is} M^\sharp-I)
\right)$$
has norm $O\left(\frac1N\right)$ on the orthogonal of the constant function (for any real $s$ if the spectrum of $M^\sharp$ is contained in $[-1+\beta', 1-\beta']\cup\{1\}$, or for $q^{2is}$ away from $-1$ if the spectrum of $M^\sharp$ is contained in $[-1, 1-\beta']\cup\{1\}$). This tells us that \eqref{e:neumann2}, and hence \eqref{e:neumann}, is $O\left(\frac1N\right)$.

\subsection{Reduction to the case $D=2$}
Let us now consider Theorem \ref{t:gen} in the case of an operator whose kernel on the tree satisfies $K(x, y)\not=0\Longrightarrow d_{\mathfrak{X}}(x, y)=D.$
Theorem \ref{t:gen}, proven in the case $D=2$, can be applied to the $(q+1)q^{D-1}$-regular graph with vertex set $V_n$ and adjacency matrix $(q+1)q^{D-1} S_D$, with the notation of \eqref{e:sk}. This implies Theorem \ref{t:gen} in the general case.

\bibliographystyle{plain}
\bibliography{biblio}

\end{document}

%% file: def.tex
\newcommand{\nwc}{\newcommand}
\nwc{\nwt}{\newtheorem}
\nwt{coro}{Corollary}
\nwt{ex}{Example}
\nwt{prop}{Proposition}
\nwt{defin}{Definition}

%font change

\nwc{\mf}{\mathbf} %Latex (as in \bf not tilted math letters)
\nwc{\blds}{\boldsymbol} %Latex 
\nwc{\ml}{\mathcal} %Latex

%greek letters

\nwc{\lam}{\lambda}
\nwc{\del}{\delta}
\nwc{\Del}{\Delta}
\nwc{\Lam}{\Lambda}
\nwc{\elll}{\ell}
%blackboard bold math

\nwc{\IA}{\mathbb{A}} %algebraic
\nwc{\IB}{\mathbb{B}} %ball
\nwc{\IC}{\mathbb{C}} %complex
\nwc{\ID}{\mathbb{D}} %Dedekind
\nwc{\IE}{\mathbb{E}} %Euklides
\nwc{\IF}{\mathbb{F}} %finite field
\nwc{\IG}{\mathbb{G}} %Gauss
\nwc{\IH}{\mathbb{H}} %Hilbert\N-subgroup
\nwc{\IN}{\mathbb{N}} %natural
\nwc{\IP}{\mathbb{P}} %prime
\nwc{\IQ}{\mathbb{Q}} %rational
\nwc{\IR}{\mathbb{R}} %real
\nwc{\IS}{\mathbb{S}} %sphere
\nwc{\IT}{\mathbb{T}} %torus
\nwc{\IZ}{\mathbb{Z}} %integers
\def\bbbone{{\mathchoice {1\mskip-4mu {\rm{l}}} {1\mskip-4mu {\rm{l}}}
{ 1\mskip-4.5mu {\rm{l}}} { 1\mskip-5mu {\rm{l}}}}}
\def\bbleft{{\mathchoice {[\mskip-3mu {[}} {[\mskip-3mu {[}}{[\mskip-4mu {[}}{[\mskip-5mu {[}}}}
\def\bbright{{\mathchoice {]\mskip-3mu {]}} {]\mskip-3mu {]}}{]\mskip-4mu {]}}{]\mskip-5mu {]}}}}
\nwc{\setK}{\bbleft 1,K \bbright}
\nwc{\setN}{\bbleft 1,\cN \bbright}
 \newcommand{\Lim}{\mathop{\longrightarrow}\limits}
%Straight (vector) bold letters

%lowercase

\nwc{\va}{{\bf a}}
\nwc{\vb}{{\bf b}}
\nwc{\vc}{{\bf c}}
\nwc{\vd}{{\bf d}}
\nwc{\ve}{{\bf e}}
\nwc{\vf}{{\bf f}}
\nwc{\vg}{{\bf g}}
\nwc{\vh}{{\bf h}}
\nwc{\vi}{{\bf i}}
\nwc{\vI}{{\bf I}}
\nwc{\vj}{{\bf j}}
\nwc{\vk}{{\bf k}}
\nwc{\vl}{{\bf l}}
\nwc{\vm}{{\bf m}}
\nwc{\vM}{{\bf M}}
\nwc{\vn}{{\bf n}}
\nwc{\vo}{{\it o}}
\nwc{\vp}{{\bf p}}
\nwc{\vq}{{\bf q}}
\nwc{\vr}{{\bf r}}
\nwc{\vs}{{\bf s}}
\nwc{\vt}{{\bf t}}
\nwc{\vu}{{\bf u}}
\nwc{\vv}{{\bf v}}
\nwc{\vw}{{\bf w}}
\nwc{\vx}{{\bf x}}
\nwc{\vy}{{\bf y}}
\nwc{\vz}{{\bf z}}
\nwc{\bal}{\blds{\alpha}}
\nwc{\bep}{\blds{\epsilon}}
\nwc{\barbep}{\overline{\blds{\epsilon}}}
\nwc{\bnu}{\blds{\nu}}
\nwc{\bmu}{\blds{\mu}}
\nwc{\bet}{\blds{\eta}}

%bold letters
%\b* letters are tilted in math mode and scale in equations. 
%but cannot be used in plain text format.

%I. lowercase

\nwc{\bk}{\blds{k}}
\nwc{\bm}{\blds{m}}
\nwc{\bM}{\blds{M}}
\nwc{\bp}{\blds{p}}
\nwc{\bq}{\blds{q}}
\nwc{\bn}{\blds{n}}
\nwc{\bv}{\blds{v}}
\nwc{\bw}{\blds{w}}
\nwc{\bx}{\blds{x}}
\nwc{\bxi}{\blds{\xi}}
\nwc{\by}{\blds{y}}
\nwc{\bz}{\blds{z}}

%caligraphic

\nwc{\cA}{\ml{A}}
\nwc{\cB}{\ml{B}}
\nwc{\cC}{\ml{C}}
\nwc{\cD}{\ml{D}}
\nwc{\cE}{\ml{E}}
\nwc{\cF}{\ml{F}}
\nwc{\cG}{\ml{G}}
\nwc{\cH}{\ml{H}}
\nwc{\cI}{\ml{I}}
\nwc{\cJ}{\ml{J}}
\nwc{\cK}{\ml{K}}
\nwc{\cL}{\ml{L}}
\nwc{\cM}{\ml{M}}
\nwc{\cN}{\ml{N}}
\nwc{\cO}{\ml{O}}
\nwc{\cP}{\ml{P}}
\nwc{\cQ}{\ml{Q}}
\nwc{\cR}{\ml{R}}
\nwc{\cS}{\ml{S}}
\nwc{\cT}{\ml{T}}
\nwc{\cU}{\ml{U}}
\nwc{\cV}{\ml{V}}
\nwc{\cW}{\ml{W}}
\nwc{\cX}{\ml{X}}
\nwc{\cY}{\ml{Y}}
\nwc{\cZ}{\ml{Z}}

\nwc{\fA}{\mathfrak{a}}
\nwc{\fB}{\mathfrak{b}}
\nwc{\fC}{\mathfrak{c}}
\nwc{\fD}{\mathfrak{d}}
\nwc{\fE}{\mathfrak{e}}
\nwc{\fF}{\mathfrak{f}}
\nwc{\fG}{\mathfrak{g}}
\nwc{\fH}{\mathfrak{h}}
\nwc{\fI}{\mathfrak{i}}
\nwc{\fJ}{\mathfrak{j}}
\nwc{\fK}{\mathfrak{k}}
\nwc{\fL}{\mathfrak{l}}
\nwc{\fM}{\mathfrak{m}}
\nwc{\fN}{\mathfrak{n}}
\nwc{\fO}{\mathfrak{o}}
\nwc{\fP}{\mathfrak{p}}
\nwc{\fQ}{\mathfrak{q}}
\nwc{\fR}{\mathfrak{r}}
\nwc{\fS}{\mathfrak{s}}
\nwc{\fT}{\mathfrak{t}}
\nwc{\fU}{\mathfrak{u}}
\nwc{\fV}{\mathfrak{v}}
\nwc{\fW}{\mathfrak{w}}
\nwc{\fX}{\mathfrak{x}}
\nwc{\fY}{\mathfrak{y}}
\nwc{\fZ}{\mathfrak{z}}

%% (wide)tilde letters

\nwc{\tA}{\widetilde{A}}
\nwc{\tB}{\widetilde{B}}
\nwc{\tE}{E^{\vareps}}
%\nwc{\tcO}{\widetilde{\mathcal{O}}}
\nwc{\tk}{\tilde k}
\nwc{\tN}{\tilde N}
\nwc{\tP}{\widetilde{P}}
\nwc{\tQ}{\widetilde{Q}}
\nwc{\tR}{\widetilde{R}}
\nwc{\tV}{\widetilde{V}}
\nwc{\tW}{\widetilde{W}}
\nwc{\ty}{\tilde y}
\nwc{\teta}{\tilde \eta}
\nwc{\tdelta}{\tilde \delta}
\nwc{\tlambda}{\tilde \lambda}
%\nwc{\tchi}{\tilde \chi}
\nwc{\ttheta}{\tilde \theta}
\nwc{\tvartheta}{\tilde \vartheta}
\nwc{\tPhi}{\widetilde \Phi}
\nwc{\tpsi}{\tilde \psi}
\nwc{\tmu}{\tilde \mu}

%miscellany
\nwc{\To}{\longrightarrow} %limits

\nwc{\ad}{\rm ad}
\nwc{\eps}{\epsilon}
\nwc{\ep}{\epsilon}
\nwc{\vareps}{\varepsilon}

\def\bom{\mathbf{\omega}}
\def\om{{\omega}}
\def\ep{\epsilon}
\def\tr{{\rm tr}}
\def\diag{{\rm diag}}
\def\Tr{{\rm Tr}}
\def\i{{\rm i}}
\def\mi{{\rm i}}
\def\e{{\rm e}}
\def\sq2{\sqrt{2}}
\def\sqn{\sqrt{N}}
\def\vol{\mathrm{vol}}
\def\defi{\stackrel{\rm def}{=}}
\def\t2{{\mathbb T}^2}
%\def\tt2{{\mathbb T}^2}
%\nwc{\t1}{{\mathbb T}^1}
\def\s2{{\mathbb S}^2}
\def\hn{\mathcal{H}_{N}}
\def\shbar{\sqrt{\hbar}}
\def\A{\mathcal{A}}
\def\N{\mathbb{N}}
\def\T{\mathbb{T}}
\def\R{\mathbb{R}}
\def\RR{\mathbb{R}}
\def\Z{\mathbb{Z}}
\def\C{\mathbb{C}}
\def\O{\mathcal{O}}
\def\Sp{\mathcal{S}_+}
\def\Lap{\triangle}
\nwc{\lap}{\bigtriangleup}
\nwc{\rest}{\restriction}
\nwc{\Diff}{\operatorname{Diff}}
\nwc{\diam}{\operatorname{diam}}
\nwc{\Res}{\operatorname{Res}}
\nwc{\Spec}{\operatorname{Spec}}
\nwc{\Vol}{\operatorname{Vol}}
\nwc{\Op}{\operatorname{Op}}
\nwc{\supp}{\operatorname{supp}}
\nwc{\Span}{\operatorname{span}}

\nwc{\dia}{\varepsilon}
\nwc{\cut}{f}
\nwc{\qm}{u_\hbar}

\def\hto0{\xrightarrow{\hbar\to 0}}
\def\htoo{\stackrel{h\to 0}{\longrightarrow}}
\def\rto0{\xrightarrow{r\to 0}}
\def\rtoo{\stackrel{r\to 0}{\longrightarrow}}
\def\ntoinf{\xrightarrow{n\to +\infty}}

\providecommand{\abs}[1]{\lvert#1\rvert}
\providecommand{\norm}[1]{\lVert#1\rVert}
\providecommand{\set}[1]{\left\{#1\right\}}

\nwc{\la}{\langle}
\nwc{\ra}{\rangle}
\nwc{\lp}{\left(}
\nwc{\rp}{\right)}

%\nwc{\bal}{\begin{align}}
\nwc{\bequ}{\begin{equation}}
\nwc{\be}{\begin{equation}}
\nwc{\ben}{\begin{equation*}}
\nwc{\bea}{\begin{eqnarray}}
\nwc{\bean}{\begin{eqnarray*}}
\nwc{\bit}{\begin{itemize}}
\nwc{\bver}{\begin{verbatim}}

%\nwc{\eal}{\end{align}}
\nwc{\eequ}{\end{equation}}
\nwc{\ee}{\end{equation}}
\nwc{\een}{\end{equation*}}
\nwc{\eea}{\end{eqnarray}}
\nwc{\eean}{\end{eqnarray*}}
\nwc{\eit}{\end{itemize}}
\nwc{\ever}{\end{verbatim}}

\newcommand{\defeq}{\stackrel{\rm{def}}{=}}

%% file: ergodicity-graphs.bbl
\def\cprime{$'$}
\begin{thebibliography}{10}

\bibitem{Alon}
N.~Alon.
\newblock Eigenvalues and expanders.
\newblock {\em Combinatorica}, 6(2):83--96, 1986.
\newblock Theory of computing (Singer Island, Fla., 1984).

\bibitem{BKS07}
G.~Berkolaiko, J.~P. Keating, and U.~Smilansky.
\newblock Quantum ergodicity for graphs related to interval maps.
\newblock {\em Comm. Math. Phys.}, 273(1):137--159, 2007.

\bibitem{BKW04}
G.~Berkolaiko, J.~P. Keating, and B.~Winn.
\newblock No quantum ergodicity for star graphs.
\newblock {\em Comm. Math. Phys.}, 250(2):259--285, 2004.

\bibitem{Bol01}
B{\'e}la Bollob{\'a}s.
\newblock {\em Random graphs}, volume~73 of {\em Cambridge Studies in Advanced
  Mathematics}.
\newblock Cambridge University Press, Cambridge, second edition, 2001.

\bibitem{BL}
Shimon Brooks and Elon Lindenstrauss.
\newblock Non-localization of eigenfunctions on large regular graphs.
\newblock {\em Israel J. of Math.}
\newblock To appear,.

\bibitem{CdV85}
Y.~Colin~de Verdi{\`e}re.
\newblock Ergodicit\'e et fonctions propres du laplacien.
\newblock {\em Comm. Math. Phys.}, 102(3):497--502, 1985.

\bibitem{CS99}
Michael Cowling and Alberto~G. Setti.
\newblock The range of the {H}elgason-{F}ourier transformation on homogeneous
  trees.
\newblock {\em Bull. Austral. Math. Soc.}, 59(2):237--246, 1999.

\bibitem{Dumitriu}
Ioana Dumitriu and Soumik. Pal.
\newblock Sparse regular random graphs: Spectral density and eigenvectors.
\newblock {\em Ann. Prob.}, 40(5):2197--2235, 2012.

\bibitem{Elon}
Yehonatan Elon.
\newblock Eigenvectors of the discrete {L}aplacian on regular graphs---a
  statistical approach.
\newblock {\em J. Phys. A}, 41(43):435203, 17, 2008.

\bibitem{ElonSmi}
Yehonatan Elon and Uzy Smilansky.
\newblock Percolating level sets of the adjacency eigenvectors of {$d$}-regular
  graphs.
\newblock {\em J. Phys. A}, 43(45):455209, 13, 2010.

\bibitem{EK11-2}
L{\'a}szl{\'o} Erd{\H{o}}s and Antti Knowles.
\newblock Quantum diffusion and delocalization for band matrices with general
  distribution.
\newblock {\em Ann. Henri Poincar\'e}, 12(7):1227--1319, 2011.

\bibitem{EK11}
L{\'a}szl{\'o} Erd{\H{o}}s and Antti Knowles.
\newblock Quantum diffusion and eigenfunction delocalization in a random band
  matrix model.
\newblock {\em Comm. Math. Phys.}, 303(2):509--554, 2011.

\bibitem{ESY09}
L{\'a}szl{\'o} Erd{\H{o}}s, Benjamin Schlein, and Horng-Tzer Yau.
\newblock Local semicircle law and complete delocalization for {W}igner random
  matrices.
\newblock {\em Comm. Math. Phys.}, 287(2):641--655, 2009.

\bibitem{ESY09-2}
L{\'a}szl{\'o} Erd{\H{o}}s, Benjamin Schlein, and Horng-Tzer Yau.
\newblock Semicircle law on short scales and delocalization of eigenvectors for
  {W}igner random matrices.
\newblock {\em Ann. Probab.}, 37(3):815--852, 2009.

\bibitem{Fri08}
Joel Friedman.
\newblock A proof of {A}lon's second eigenvalue conjecture and related
  problems.
\newblock {\em Mem. Amer. Math. Soc.}, 195(910):viii+100, 2008.

\bibitem{Gnu10}
S.~Gnutzmann, J.~P. Keating, and F.~Piotet.
\newblock Eigenfunction statistics on quantum graphs.
\newblock {\em Ann. Physics}, 325(12):2595--2640, 2010.

\bibitem{JakMilRivR}
Dmitry Jakobson, Stephen~D. Miller, Igor Rivin, and Ze{\'e}v Rudnick.
\newblock Eigenvalue spacings for regular graphs.
\newblock In {\em Emerging applications of number theory ({M}inneapolis, {MN},
  1996)}, volume 109 of {\em IMA Vol. Math. Appl.}, pages 317--327. Springer,
  New York, 1999.

\bibitem{KMW03}
J.~P. Keating, J.~Marklof, and B.~Winn.
\newblock Value distribution of the eigenfunctions and spectral determinants of
  quantum star graphs.
\newblock {\em Comm. Math. Phys.}, 241(2-3):421--452, 2003.

\bibitem{Kes59}
Harry Kesten.
\newblock Symmetric random walks on groups.
\newblock {\em Trans. Amer. Math. Soc.}, 92:336--354, 1959.

\bibitem{KotSmi97}
Tsampikos Kottos and Uzy Smilansky.
\newblock Quantum chaos on graphs.
\newblock {\em Phys. Rev. Lett.}, 79:4794--7, 1997.

\bibitem{KotSmi99}
Tsampikos Kottos and Uzy Smilansky.
\newblock Periodic orbit theory and spectral statistics for quantum graphs.
\newblock {\em Ann. Physics}, 274(1):76--124, 1999.

\bibitem{LR96}
John~D. Lafferty and Daniel~N. Rockmore.
\newblock Level spacings for {C}ayley graphs.
\newblock In {\em Emerging applications of number theory ({M}inneapolis, {MN},
  1996)}, volume 109 of {\em IMA Vol. Math. Appl.}, pages 373--386. Springer,
  New York, 1999.

\bibitem{LM}
Etienne Le~Masson.
\newblock Pseudo-differential calculus on homogeneous trees.
\newblock {\em Preprint}, 2013.
\newblock arXiv:1302.5387.

\bibitem{LPS88}
A.~Lubotzky, R.~Phillips, and P.~Sarnak.
\newblock Ramanujan graphs.
\newblock {\em Combinatorica}, 8(3):261--277, 1988.

\bibitem{MSS}
Adam Marcus, Daniel~A. Spielman, and Nikhil Srivastava.
\newblock Interlacing families i: Bipartite ramanujan graphs of all degrees.
\newblock preprint 2013.

\bibitem{McK81}
Brendan~D. McKay.
\newblock The expected eigenvalue distribution of a large regular graph.
\newblock {\em Linear Algebra Appl.}, 40:203--216, 1981.

\bibitem{MKWW}
Brendan~D. McKay, Nicholas~C. Wormald, and Beata Wysocka.
\newblock Short cycles in random regular graphs.
\newblock {\em Electron. J. Combin.}, 11(1):Research Paper 66, 12 pp.
  (electronic), 2004.

\bibitem{MF91}
A.~D. Mirlin and Yan~V. Fyodorov.
\newblock Universality of level correlation function of sparse random matrices.
\newblock {\em J. Phys. A}, 24(10):2273--2286, 1991.

\bibitem{Pinsker}
M.~S. Pinsker.
\newblock On the complexity of a concentrator.
\newblock {\em 7th International Teletraffic Conference}, pages 318/1--318/4,
  1973.

\bibitem{Smi07}
Uzy Smilansky.
\newblock Quantum chaos on discrete graphs.
\newblock {\em J. Phys. A}, 40(27):F621--F630, 2007.

\bibitem{Smi10}
Uzy Smilansky.
\newblock Discrete graphs -- a paradigm model for quantum chaos.
\newblock {\em S\'eminaire Poincar\'e}, XIV:1--26, 2010.

\bibitem{Sni}
A.~I. {\v{S}}nirel{\cprime}man.
\newblock Ergodic properties of eigenfunctions.
\newblock {\em Uspehi Mat. Nauk}, 29(6(180)):181--182, 1974.

\bibitem{Ter99}
Audrey Terras.
\newblock {\em Fourier analysis on finite groups and applications}, volume~43
  of {\em London Mathematical Society Student Texts}.
\newblock Cambridge University Press, Cambridge, 1999.

\bibitem{TranVu}
Linh~V. Tran, Van~H. Vu, and Ke~Wang.
\newblock Sparse random graphs: eigenvalues and eigenvectors.
\newblock {\em Random Structures Algorithms}, 42(1):110--134, 2013.

\bibitem{Zel86}
Steven Zelditch.
\newblock Pseudodifferential analysis on hyperbolic surfaces.
\newblock {\em J. Funct. Anal.}, 68(1):72--105, 1986.

\bibitem{Zel87}
Steven Zelditch.
\newblock Uniform distribution of eigenfunctions on compact hyperbolic
  surfaces.
\newblock {\em Duke Math. J.}, 55(4):919--941, 1987.

\bibitem{ZelC}
Steven Zelditch.
\newblock Quantum ergodicity of {$C^*$} dynamical systems.
\newblock {\em Comm. Math. Phys.}, 177(2):507--528, 1996.

\end{thebibliography}
